\documentclass[journal]{IEEEtran}
\usepackage{graphicx}
\usepackage{amsmath}

\usepackage[noend]{algpseudocode}
\usepackage{algorithmicx,algorithm}
\usepackage{amsthm}
\usepackage{amsfonts}
\usepackage {subfigure}

\usepackage{color}
\usepackage{cite}
\usepackage{setspace}
\newtheorem{Lemma}{Lemma}
\newtheorem{Theorem}{Theorem}

\newtheorem{Proposition}{Proposition}
\newtheorem{Definition}{Definition}
\newtheorem{Assumption}{Assumption}
\usepackage{amssymb}
\usepackage{color}
\graphicspath{{images/}}
\usepackage{booktabs}
\usepackage{multirow} 
\usepackage{makecell}
\usepackage[numbers,sort&compress]{natbib}
\newtheorem{remark}{Remark}
\usepackage{stfloats}
\usepackage{amsmath}
\usepackage{amssymb}
\usepackage{enumitem}
\usepackage{threeparttable}
\usepackage{comment}
\usepackage{bm}
\usepackage{colortbl}
\usepackage{tabularx}
\usepackage[table]{xcolor}
\setlist{noitemsep, topsep=2pt, parsep=0pt, partopsep=0pt}
\begin{document}	
	\title{Vorticity Dissipation Based Routing: A Fluid-Kinetic Framework for Loop-Free Transport in Ultra-Dense Networks}
	
	\author{Wen-Yu Dong, Weiwei Jiang,~\IEEEmembership{Senior Member,~IEEE}, Song Zhao,  Rui-Si Han,\\ Qi Bi,~\IEEEmembership{Fellow,~IEEE}, Sheng Chen,~\IEEEmembership{Life Fellow,~IEEE}	%
		
		\thanks{W.-Y. Dong, S. Zhao and Q. Bi are with Future Technology Research Center, China Telecom Research Institute, Beijing 102209, China (E-mails: dongwy@chinatelecom.cn; zhaosong1@chinatelecom.cn; qibi@chinatelecom.cn).} %
		\thanks{W. Jiang is with the School of Information and Communication Engineering, Beijing University of Posts and Telecommunications, Beijing, 100876, China (E-mail: jww@bupt.edu.cn).}
		\thanks{R.-S. Han is with Cloud Network Operating System R\&D Center, China Telecom, Beijing 102209, China (E-mail: hanruisi@chinatelecom.cn).} %
		\thanks{S. Chen is with the School of Electronics and Computer Science, University of Southampton, Southampton SO17 1BJ, U.K., and also with  Faculty of Information Science and Technology, Ocean University of China, Qingdao 266100, China (E-mail: sqc@ecs.soton.ac.uk).} %
		\thanks{\scriptsize Digital Object Identifier: 10.1109/TMC.2026.3722828. 
			\copyright~2026 IEEE. Personal use of this material is permitted.
			Permission from IEEE must be obtained for all other uses, in any current or future media,
			including reprinting/republishing this material for advertising or promotional purposes,
			creating new collective works, for resale or redistribution to servers or lists,
			or reuse of any copyrighted component of this work in other works.
		}
		\vspace*{-5mm}
	}
	
	
	
	\maketitle 
	\vspace{-1.5cm}
	\begin{abstract}
		Discrete routing protocols in ultra-dense wireless networks are constrained by signaling overhead and transient routing loops that degrade {\color{black}radio-resource efficiency}. While continuum modeling provides a scalable alternative, existing scalar density approaches lack the vector geometric structure required to characterize these topological anomalies. This paper introduces a fluid-kinetic framework, vorticity dissipation-based routing (VDR), utilizing the Helmholtz-Hodge decomposition. We demonstrate that the macroscopic traffic flux can be orthogonally decoupled into a demand-driven irrotational component and a loop-induced solenoidal component representing routing vorticity. Building on this insight, we define network vorticity as a macroscopic metric to quantify topological inefficiency. Routing optimization is formulated as a gradient flow on an enstrophy functional, yielding a vorticity dissipation equation as the governing dynamic law. Lyapunov stability analysis proves that this mechanism ensures the monotonic decay of global enstrophy toward an asymptotically loop-free equilibrium. Numerical results validate that {\color{black}VDR suppresses realized forwarding loops, reduces end-to-end delay, maintains robust packet delivery, and exhibits near-linear scaling under fixed-area densification while explicitly accounting for the grid-dependent Poisson-solver cost.}
	\end{abstract}
	
	\begin{IEEEkeywords}
		Ultra-dense networks, continuum modeling, Helmholtz-Hodge decomposition, network vorticity, loop-free routing.
	\end{IEEEkeywords}
	
	\section{Introduction} \label{S1}
	
	The ultra-dense deployment of Internet of Things (IoT) devices and small cells in 6G networks precipitates a fundamental topological shift from centralized architectures to multi-hop mesh networking, creating quasi-continuous media \cite{IoTJDong}. In this regime, the explosive growth of node density renders microscopic state maintenance computationally intractable, necessitating a scalable macroscopic paradigm to ensure efficient information transport \cite{Polese2020}. Consequently, developing routing mechanisms that can handle this phase transition while maintaining topological stability is a critical challenge for emerging architectures like massive machine-type communications and integrated access and backhaul.
	
	However, managing discrete packet flows becomes unstable as network density increases, where conventional protocols encounter severe scalability limitations \cite{Kim2017RPLSurvey}. A primary bottleneck in decentralized routing is the formation of transient routing loops. These cyclic paths, caused by stale local states or conflicting objectives, significantly deplete bandwidth and exacerbate energy hole phenomena. Unlike physical obstacles that define static boundaries, these topological anomalies are dynamic and distributed, making them difficult to rectify using traditional microscopic strategies without incurring prohibitive signaling overhead and latency.
	
	To address this, existing research primarily follows three technical routes: discrete graph protocols, such as ad hoc on-demand distance vector (AODV) and IPv6 routing protocol for low-power and lossy networks (RPL), stochastic backpressure control, and learning-based inference, all of which struggle with either scalability bottlenecks or a lack of deterministic safety guarantees. While continuum modeling offers a fourth route by treating traffic as a fluid to achieve scalability \cite{Lighthill1955}, current scalar approaches fail to capture the vector geometric structure of rotational loops. Therefore, this paper adopts a fluid-kinetic framework based on the Helmholtz-Hodge decomposition (HHD) \cite{Griffiths1999}. We choose this route because it mathematically decouples traffic into an irrotational component representing effective transport and a solenoidal component representing redundant circulation, enabling the elimination of loops through a vorticity dissipation process that guarantees asymptotic stability independent of network density.
	
	\subsection{Related Work} \label{S1.1}
	
	Existing routing protocols are generally categorized into microscopic approaches, which model interactions on discrete graphs, and macroscopic approaches, which abstract ultra-dense networks as continuum fields.
	
	\subsubsection{Discrete Graph-Based Approaches} 
	Traditional routing frameworks abstract the network as a discrete graph to maximize path efficiency. Classical topology-based protocols, such as the reactive AODV \cite{Perkins1999AODV}, the proactive optimized link state routing (OLSR) \cite{Clausen2003OLSR}, and the RPL \cite{Kim2017RPLSurvey}, rely on topological metrics and rigid graph structures. To reduce global maintenance overhead, geographic strategies, like greedy perimeter stateless routing (GPSR) \cite{Karp2000GPSR}, utilize nodal coordinates for localized forwarding. To enhance adaptability in highly dynamic scenarios, recent research has also explored various bio-inspired heuristics, e.g., ant colony optimization for flying ad hoc networks \cite{Wang2025TVT}, and topology-aware relay selection methods \cite{Yang2025TVT}. Such highly dynamic scenarios are also expected in emerging low-altitude wireless networks with integrated sensing and communication~\cite{Ding2026ISACWaveform}. However, these discrete and heuristic methods are fundamentally constrained by graph granularity: they often lack deterministic safety guarantees and scale poorly in ultra-dense networks, as the signaling and computational overhead required to maintain routing invariants grows prohibitively.
	
	\subsubsection{Stochastic Optimization and Stability Control} 
	Backpressure routing \cite{Tassiulas1992} and its delay-reduction variants \cite{Bui2011} achieve theoretical throughput optimality via Lyapunov drift minimization based on queue differentials. To constrain the directionality of inefficient random walks inherent to low-gradient scenarios, recent advancements have integrated various structural priors and adaptive bias tuning into the backpressure framework, with applications ranging from low Earth orbit (LEO) satellite networks \cite{Deng2023DBPR} to wireless edge computing \cite{Zhao2024BiasedBP, Mahfujul2024TNSE}. Nevertheless, these queue-based mechanisms remain reactive. In high-mobility scenarios, the lag between topology changes and queue convergence frequently induces transient micro-loops, while the complex coupling dynamics incur computational costs that conflict with the real-time constraints of IoT systems.
	
	\subsubsection{Learning-Based Routing} 
	Autonomous network control increasingly employs learning-based paradigms to infer adaptive forwarding policies. Methods ranging from reinforcement learning (RL) \cite{Arafat2022QL} and deep RL (DRL) \cite{Xiao2021DRLSurvey} to graph neural networks (GNNs) \cite{Rusek2020, Shen2021GNN} have been integrated to capture complex spatial dependencies and non-linear interactions in irregular wireless topologies. While highly adaptable, these data-driven models typically lack deterministic safety guarantees \cite{Khan2024}. The probabilistic nature of neural inference creates generalization risks on out-of-distribution (OoD) topologies, and the dependency on extensive retraining imposes significant overhead. Crucially, learning-based formulations generally fail to provide strict loop-freedom proofs, suffering from potential transient oscillations during stochastic exploration that necessitate analytical frameworks with rigorous bounds.
	
	\subsubsection{Topological Signal Processing and Discrete Hodge Theory} 
	The integration of algebraic topology into network science has established the field of topological signal processing (TSP) \cite{Barbarossa2020TSP}. By leveraging discrete exterior calculus, existing research applies the graph Hodge theorem to analyze edge flows, effectively modeling random walks \cite{Schaub2020RandomWalks} and resolving cyclic inconsistencies \cite{Jiang2011StatisticalRanking}. However, applying this directly to dynamic routing presents fundamental challenges: computing the discrete curl requires constructing higher-order simplicial complexes, such {\color{black}as tracking} all 2-simplices, incurring massive combinatorial overhead. Furthermore, discrete Hodge theory functions primarily as a retrospective tool for static flows rather than a low-complexity real-time control mechanism.
	
	\subsubsection{Continuum Field Models} 
	To address scalability, macroscopic approaches model ultra-dense networks as continua using calculus-based tools \cite{Jacquet2004, Toumpis2003Continuum, Samarakoon2016MFG}. A critical limitation of most continuum formulations is the assumption of irrotational potential flow, which enforces a zero-curl condition and explicitly excludes the rotational component of the vector field. Since routing loops correspond mathematically to circulation, curl-based operators are necessary to quantify and regulate these effects. This gap motivates a unified framework that explicitly models and dissipates routing vorticity to ensure loop-free transport.

	\begin{table*}[!t]
		\vspace*{-1mm}
		\scriptsize
		\centering
		\caption{Comparison of capabilities for various routing paradigms}
		\label{tab:comparison} 
		\vspace*{-2mm}
		\renewcommand{\arraystretch}{1.2}
		\begin{threeparttable}
			\resizebox{\textwidth}{!}{
				\begin{tabular}{l c c c c c}
					\toprule
					\textbf{Paradigm \& Representative Works} & \textbf{Scalability} & \textbf{Loop-Freedom Guarantee} & \textbf{Congestion Awareness} & \textbf{Deterministic Safety} & \textbf{Low Latency} \\
					\midrule
					\textbf{Discrete Protocols} & \multirow{2}{*}{$-$} & \multirow{2}{*}{$\checkmark$} & \multirow{2}{*}{$-$} & \multirow{2}{*}{$\checkmark$} & \multirow{2}{*}{$-$} \\
					(e.g., AODV \cite{Perkins1999AODV}, RPL \cite{Kim2017RPLSurvey}, GPSR \cite{Karp2000GPSR}) & & & & & \\
					\midrule
					
					\textbf{Stochastic Control} & \multirow{2}{*}{$-$} & \multirow{2}{*}{$-$} & \multirow{2}{*}{$\checkmark$} & \multirow{2}{*}{$-$} & \multirow{2}{*}{$-$} \\
					(e.g., Backpressure \cite{Tassiulas1992, Bui2011}) & & & & & \\
					\midrule
					
					\textbf{Learning-Based} & \multirow{2}{*}{$\checkmark$} & \multirow{2}{*}{$-$} & \multirow{2}{*}{$\checkmark$} & \multirow{2}{*}{$-$} & \multirow{2}{*}{$\checkmark$} \\
					(e.g., QTAR \cite{Arafat2022QL}, DRL \cite{Xiao2021DRLSurvey}, GNN \cite{Shen2021GNN}) & & & & & \\
					\midrule
					
					\textbf{Discrete Exterior Calculus (DEC)} & \multirow{2}{*}{$-$} & \multirow{2}{*}{$\checkmark$} & \multirow{2}{*}{$-$} & \multirow{2}{*}{$\checkmark$} & \multirow{2}{*}{$-$} \\
					(e.g., Graph Hodge / TSP \cite{Barbarossa2020TSP, Schaub2020RandomWalks}) & & & & & \\
					\midrule
					
					\textbf{Traditional Continuum} & \multirow{2}{*}{$\checkmark$} & \multirow{2}{*}{$-$} & \multirow{2}{*}{$-$} & \multirow{2}{*}{$\checkmark$} & \multirow{2}{*}{$\checkmark$} \\
					(e.g., Potential fields \cite{Toumpis2003Continuum}) & & & & & \\
					
					\midrule
					\textbf{Proposed VDR} &$\checkmark$ &$\checkmark$ &$\checkmark$ &$\checkmark$&$\checkmark$ \\
					\bottomrule
				\end{tabular}
			}
			\begin{tablenotes}
				\scriptsize
				\item  $\checkmark$: Supported / High Performance; $-$: Not Considered / Low Performance.
				\item	\textbf{Abbreviations}: AODV: Ad hoc On-Demand Distance Vector; RPL: IPv6 Routing Protocol for Low-Power and Lossy Networks; GPSR: Greedy Perimeter Stateless Routing; QTAR: Q-learning-based Topology-Aware Routing; GNN: Graph Neural Networks; DRL: Deep Reinforcement Learning; DEC: Discrete Exterior Calculus; TSP: Topological Signal Processing; VDR: Vorticity Dissipation Based Routing (Proposed).
			\end{tablenotes}
		\end{threeparttable}
		\vspace{-5mm}
	\end{table*}
	
	\subsection{Contributions and Organization} \label{subsec:contributions} 
	
	To bridge the gap between microscopic exactness and macroscopic scalability, this paper proposes a continuum routing framework termed vorticity dissipation based routing (VDR). By introducing the HHD, the wireless traffic field is decoupled into orthogonal irrotational and solenoidal components. This transformation converts {\color{black}loop-prone routing behavior from a purely combinatorial graph phenomenon into a field-level circulation-dissipation problem}. The specific contributions of this work are summarized as follows.
	
	\begin{itemize}
		\item \textit{Continuum vector field modeling and network vorticity definition:} {\color{black}A mapping methodology based on kernel density estimation (KDE) is established to abstract observable nodal states into continuous traffic density, source-intensity, and flux fields. Within this continuum representation, network vorticity is defined as a macroscopic measure of local loop tendency. We further clarify the finite-density interpretation of this continuum approximation and analyze how graph resolution, KDE bandwidth, and finite-node sampling affect the continuum-to-graph projection.}
		
		\item \textit{Analytical decomposition and energy interpretation:} {\color{black}The scalar and vector potentials of the traffic field are derived by solving the associated Poisson equations, and the traffic flux is decomposed into demand-driven irrotational and loop-induced solenoidal components. The associated orthogonality property shows that redundant circulation can be separated from useful source-to-sink transport. We also provide a discrete radio-resource interpretation, clarifying that the transport-energy and solenoidal-energy functionals should be viewed as normalized forwarding-load surrogates rather than complete physical battery-energy models.}
		
		\item \textit{Vorticity dissipation control and practical discrete implementation:} 
		{\color{black}
			Routing optimization is formulated as a gradient flow on the enstrophy functional, yielding a vorticity dissipation law with Lyapunov-based monotonic decay under the stated boundary conditions. To enable finite-graph deployment, we develop a distributed realization using graph Laplacian operators, normalized forwarding utilities, adaptive source-intensity estimation, and packet-level Bloom-filter loop guarding. The implementation further accounts for forwarding delay, packet drops, memory overhead, and grid-dependent computational complexity.
		}
	\end{itemize}
	From a broader theoretical perspective, the proposed VDR framework can be viewed as a routing-oriented application of the fluid-spatiotemporal stochastic geometry (F-STSG) framework for modeling information flow in non-stationary fields~\cite{Dong2026FSTSG}. Relatedly, the fluid-dynamic paradigm has also been applied to flux-aware infrastructure provisioning in UAV logistics networks~\cite{Dong2026TMC}.

	To position this work within the literature, Table~\ref{tab:comparison} compares the proposed VDR with conventional routing paradigms. The remainder of this paper is organized as follows. Section~\ref{S2} presents the system model and the construction of the continuous traffic field. Section~\ref{S3} details the analytical resolution of the Helmholtz potentials and the derivation of the energy orthogonality. Section~\ref{S4} establishes the dynamic evolution laws, {\color{black}finite-density interpretation,} and corresponding discrete implementation algorithms. Section~\ref{S5} provides numerical performance evaluations, and Section~\ref{S6} concludes the paper. The key notations and symbols used throughout this paper are summarized in Table~\ref{tab:notations}.
	
	\begin{table}[!t]
		\vspace*{-1mm}
		\scriptsize
		\caption{Summary of Key Mathematical Notations}
		\vspace*{-2mm}
		\label{tab:notations} 
		\centering
		\renewcommand{\arraystretch}{1.1} 
		\setlength{\tabcolsep}{3pt}       
		\resizebox{\columnwidth}{!}{
			\begin{tabular}{c|l}
				\hline
				\textbf{Symbol} & \textbf{Definition / Physical Interpretation} \\
				\hline
				$\mathbf{\Phi}$ & Traffic flux density vector field (information flow) \\
				$\rho$ & Traffic density field (spatial congestion level) \\
				$\mathbf{\Phi}_{\mathrm{irr}}$ & Irrotational flux component (effective transport) \\
				$\mathbf{\Phi}_{\mathrm{sol}}$ & Solenoidal flux component (routing vorticity/loops) \\
				$\psi, \mathbf{A}$ & Scalar potential and Vector potential fields \\
				$\boldsymbol{\omega}$ & Network vorticity field ($\nabla \times \mathbf{\Phi}$) \\
				
				$\mathcal{E}_{\mathrm{tot}}$ & Total transport energy ($\int \|\mathbf{\Phi}\|^2$) \\
				$\mathcal{E}_{\mathrm{sol}}$ & Solenoidal transport energy (Energy of loops) \\
				$\eta_{\mathrm{vor}}$ & Vortex energy ratio (VER), $\mathcal{E}_{\mathrm{sol}} / \mathcal{E}_{\mathrm{tot}}$ \\
				$\mathcal{F}[\mathbf{\Phi}], \mathcal{E}_{\mathrm{vor}}$ & Network enstrophy functional ($\frac{1}{2} \int \|\boldsymbol{\omega}\|^2$) \\
				
				$\tau$ & Virtual optimization time for field evolution \\
				$\mathbf{L}$ & Graph Laplacian matrix ($\mathbf{B} \mathbf{W} \mathbf{B}^{\top}$) \\
				{\color{black}$\alpha_i$} & {\color{black}Adaptive weight balancing normalized potential drive and normalized queue pressure} \\
				{\color{black}$Q_{\max}$} & {\color{black}Queue capacity used for queue normalization} \\
				{\color{black}$\widehat{\Delta\psi}_{ij}$} & {\color{black}Normalized potential-difference term from node $i$ to node $j$} \\
				{\color{black}$\widehat{\Delta Q}_{ij}$} & {\color{black}Normalized queue-pressure difference from node $i$ to node $j$} \\
				{\color{black}$U_{ij}$} & {\color{black}Normalized forwarding utility metric for candidate link $(i,j)$} \\
				$\Lambda(j, \mathcal{H}_p)$ & {\color{black}Hard loop-avoidance barrier for packet trajectory history} \\
				$\mathbf{B}_p$ & Bloom filter vector encoding trajectory history \\
				\hline
			\end{tabular}
		}
		\vspace*{-1mm}
	\end{table}

	\section{System Model and Problem Formulation}\label{S2}
	
	In this section, we establish the mathematical foundation for the proposed fluid-kinetic routing framework. In contrast to traditional discrete graph abstractions, we adopt a continuum perspective to model the collective transport of information in ultra-dense networks. This approach treats packet flows as a compressible fluid, enabling the application of vector field theory to quantify topological routing anomalies.
	
	\begin{figure}[!t]
			\vspace*{-1mm}
		\begin{center}
			\includegraphics[width=1\columnwidth]{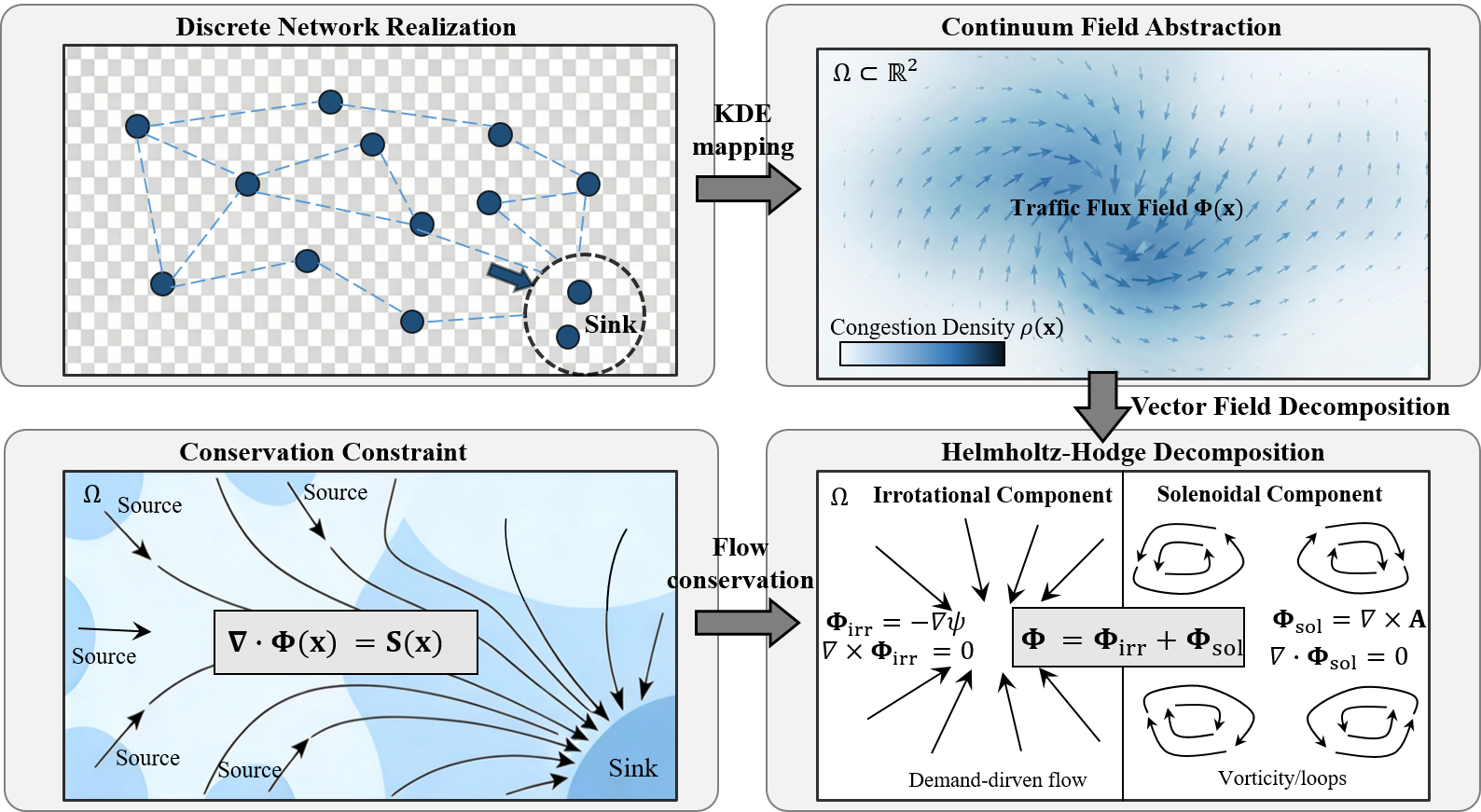}
		\end{center}
		\vspace*{-5mm}
		\caption{System model of VDR: from discrete wireless mesh to continuum traffic fields and Helmholtz-Hodge decomposition.}
		\label{fig:system_model} 
		\vspace*{-2mm}
	\end{figure}
	
	\subsection{Continuum Representation of Wireless Traffic}\label{S2.1}
	
	The proposed VDR model is depicted in Fig.~\ref{fig:system_model}.
	We consider a large-scale wireless mesh network deployed over a bounded spatial domain $\Omega \subseteq \mathbb{R}^2$ with a boundary $\partial \Omega$. The network consists of a set of static nodes $\mathcal{N}$ distributed according to a spatial point process \cite{DongTCOM, DongJSAC}. Two nodes can establish a direct wireless communication link if their Euclidean distance is no greater than the maximum communication range, denoted by $R_{\mathrm{c}}$. To characterize the macroscopic routing behavior, we abstract the discrete packet transmissions as a continuous fluid flow. 
	Let $\mathbf{x} = [x, y]^{\mathrm{T}} \in \Omega$ denote the spatial coordinates and $t \in \mathbb{R}^+$ represent time. To bridge the gap between discrete nodal states and continuous field theory, we construct the macroscopic state variables using KDE on strictly observable metrics. Unlike standard fluid models where flux is derived from velocity as $\mathbf{\Phi} = \rho \mathbf{v}$, we define the wireless traffic state through three coupled fields.
	
	\subsubsection{Traffic Density Field ($\rho$)} 
	The density field $\rho(\mathbf{x}, t)$ represents the spatial distribution of congestion potential, mapped from the normalized nodal queue length $q_i(t)\! =\! Q_i(t) / Q_{\max}$, where $Q_i(t)$ is the instantaneous queue length of node $i$ at time $t$, and $Q_{\max}$ is the maximum queue capacity of the node.  It acts as the diffusive pressure source in the forwarding plane:
	\begin{equation} 
		\rho(\mathbf{x}, t) \triangleq \sum_{i \in \mathcal{N}} q_i(t) \mathcal{K}_h\left( \| \mathbf{x} - \mathbf{x}_i \| \right),
		\label{eq:kde_density}
	\end{equation}
	where $\mathbf{x}_i \in \Omega$ denotes the spatial coordinates of node $i$, and $\mathcal{K}_h(r)\! =\! \frac{1}{2\pi h^2} \exp\left(-\frac{r^2}{2h^2}\right)$ is the Gaussian smoothing kernel with bandwidth $h$.
	
	\subsubsection{Source Intensity Field ($S$)} 
	The source field $S(\mathbf{x}, t)$ represents the net demand topology, mapped from the exponential moving average (EMA) of the packet generation/consumption rate $s_i(t)$. This field drives the global irrotational flow:
	\begin{equation} 
		S(\mathbf{x}, t) \triangleq \sum_{i \in \mathcal{N}} s_i(t) \mathcal{K}_h\left( \| \mathbf{x} - \mathbf{x}_i \| \right).
		\label{eq:kde_source}
	\end{equation}
	
	\subsubsection{Traffic Flux Field ($\mathbf{\Phi}$)} 
	The flux density vector $\mathbf{\Phi}(\mathbf{x}, t)$ is defined as the effective transport of information. Rather than being an algebraic product of density and velocity, $\mathbf{\Phi}$ is treated as the primary control variable governed by the HHD \cite{Griffiths1999}:
	\begin{equation} 
		\mathbf{\Phi}(\mathbf{x}, t) = \mathbf{\Phi}_{\mathrm{irr}}(\mathbf{x}, t) + \mathbf{\Phi}_{\mathrm{sol}}(\mathbf{x}, t), 
		\label{eq:flux_decomp_def}
	\end{equation}
	where the time-varying irrotational component $\mathbf{\Phi}_{\mathrm{irr}}(\mathbf{x}, t)$ is determined by the demand $S(\mathbf{x}, t)$ via the Poisson equation, and the solenoidal component $\mathbf{\Phi}_{\mathrm{sol}}(\mathbf{x}, t)$ represents the loop currents to be dissipated.  
	
	{\color{black}
		In networking terms, $\rho$ can be viewed as a congestion map, $S$ as a source-sink demand map, and $\mathbf{\Phi}$ as the macroscopic forwarding-flow map.}
	{\color{black}
		The bandwidth parameter $h$ dictates the spatial resolution of the continuum approximation. According to the standard bandwidth selection principles in spatial statistics~\cite{Silverman1986}, $h$ should be scaled with the characteristic interaction distance of the discrete system to balance field smoothness and topological fidelity. In this work, we parameterize the bandwidth as $h=\beta_h R_c$, where $\beta_h$ controls the smoothing scale. A smaller $\beta_h$ preserves more localized routing features but may amplify sampling noise, whereas a larger $\beta_h$ improves field regularity but may blur small-scale voids or isolated link failures. Unless otherwise specified, we use $\beta_h=0.5$, i.e., $h=R_c/2$, as the default setting. This value is not assumed to be universally optimal; its impact on vorticity localization and routing performance is evaluated through the bandwidth sensitivity analysis in Section~\ref{subsec:bandwidth_sensitivity}.
	}
	
	Under this formalism, the hydrodynamic variables map directly to key network performance metrics.
	Specifically, by decomposing the macroscopic flux field $\mathbf{\Phi}(\mathbf{x}, t)$ into its magnitude and direction, we provide the following communication-theoretic interpretations.
	\begin{itemize}
		\item The density field $\rho(\mathbf{x}, t)$ represents the spatial distribution of traffic congestion. A local maximum in $\rho$ corresponds to a backlog accumulation or a potential bottleneck.
		\item The flux magnitude $\|\mathbf{\Phi}(\mathbf{x}, t)\|$ corresponds to the effective macroscopic throughput or traffic load traversing a local neighborhood.
		\item The flux direction $\frac{\mathbf{\Phi}(\mathbf{x}, t)}{\|\mathbf{\Phi}(\mathbf{x}, t)\|}$ (or orientation) indicates the collective routing decision, guiding the information flow towards the sinks.
	\end{itemize}
	
	To ensure the theoretical consistency of the proposed framework, we adopt the following standard assumptions regarding the network topology and density limits.
	
	\begin{Assumption} \label{asm:domain_topology}  
		The network domain $\Omega$ is modeled as the simply connected convex hull of the deployment area. 
		Physical obstacles or coverage voids are treated not as topological holes (which would render $\Omega$ multiply connected and necessitate a harmonic flow component), but as regions with vanishing traffic density ($\rho(\mathbf{x}, t) \to 0$) or infinite flow resistance. 
		This ``effective medium'' approach preserves the validity of the two-term HHD (\ref{eq:flux_decomp_def}) while physically forcing the flux $\mathbf{\Phi}$ to bypass void regions via the gradient potential.
	\end{Assumption}
	
	\begin{Assumption}\label{asm:continuum_limit} 
		\textit{Continuum Limit:} The node density $\lambda$ is sufficiently high such that the mean inter-nodal distance is negligible relative to the characteristic length of $\Omega$. As established in dense network theory \cite{Toumpis2003Continuum}, under this condition, the discrete packet-hopping process asymptotically converges to a macroscopic drift-diffusion process, justifying the use of partial differential equations for global traffic modeling.
	\end{Assumption}
	
	\begin{Assumption}\label{asm:solvability} 
		\textit{Solvability Condition:} To guarantee the well-posedness of the Poisson equation under Neumann boundary conditions, the source term $S(\mathbf{x}, t)$, which represents the net traffic generation rate, must satisfy the global compatibility constraint:
		\begin{equation} 
			\int_{\Omega} S(\mathbf{x}, t)\,\mathrm{d}\mathbf{x} = 0.
		\end{equation}
		This condition implies a balanced traffic pattern where the total injection rate equals the total absorption rate, complementing the topological constraints in Assumption \ref{asm:domain_topology} to ensure the existence of a valid potential field.
	\end{Assumption}
	
	\begin{Assumption}\label{asm:timescale}
		Let $\tau_{\mathrm{net}}$ denote the characteristic timescale of network dynamics (e.g., node mobility or traffic pattern shifts $S(\mathbf{x}, t)$), and $\tau_{\mathrm{conv}}$ denote the relaxation time of the vorticity dissipation process.
		We assume that the system operates in a \textit{singular perturbation} regime where $\epsilon \triangleq \tau_{\mathrm{conv}} / \tau_{\mathrm{net}} \ll 1$. 
		Under this condition, the macroscopic flux field $\mathbf{\Phi}(\mathbf{x}, t)$ is considered to be in a \textit{quasi-static equilibrium} relative to the microscopic packet forwarding processes, satisfying $\partial \rho / \partial t \approx 0$ within each optimization window $W_t$.
	\end{Assumption}
	
	In this continuum model, the information fluid is treated as a compressible medium. Unlike incompressible physical fluids, the divergence of the flux $\nabla \cdot \mathbf{\Phi}$ is generally non-zero, reflecting the local injection (source) and absorption (sink) of data traffic within the network.
	
	\subsection{Traffic Conservation and Dynamics}\label{S2.2}
	
	The transport of information bits within the domain $\Omega$ is governed by the principle of \textit{flow conservation}. For an arbitrary control volume $\mathcal{V} \subseteq \Omega$, the rate of change of the information backlog within $\mathcal{V}$ is balanced by the net flux crossing the boundary $\partial \mathcal{V}$ and the internal traffic generation. The differential form of this conservation law is expressed by the continuity equation:
	\begin{equation} 
		\frac{\partial \rho(\mathbf{x}, t)}{\partial t} + \nabla \cdot \mathbf{\Phi}(\mathbf{x}, t) = S(\mathbf{x}, t),
		\label{eq:continuity}
	\end{equation}
	where $\nabla \cdot (\cdot)$ denotes the divergence operator, and $S(\mathbf{x}, t)$ represents the source intensity field. The intensity of $S(\mathbf{x}, t)$ characterizes the local net traffic injection rate, where positive values correspond to traffic sources, such as packet generation, and negative values correspond to traffic sinks, such as packet reception at destinations.
	
	In wireless networks, macroscopic traffic patterns typically evolve at a timescale much slower than that of individual packet transmissions. Within a sufficiently short observation window $\Delta t$, the network can be modeled as being in a quasi-stationary state, i.e., $\partial \rho / \partial t \approx 0$. Under this condition, the temporal evolution term vanishes, reducing \eqref{eq:continuity} to the Poisson-type divergence equation:
	\begin{equation} 
		\nabla \cdot \mathbf{\Phi}(\mathbf{x}) = S(\mathbf{x}).
		\label{eq:poisson_traffic}
	\end{equation}
	Equation \eqref{eq:poisson_traffic} implies that the divergence of the traffic flux density is determined strictly by the local source-sink intensity. From a fluid-dynamic perspective, the information flow is treated as a compressible fluid, where the local compression or expansion of the flux field accounts for the net traffic demand. Crucially, while \eqref{eq:poisson_traffic} imposes a strict constraint on the divergent component of the flow, the solenoidal (divergence-free) component remains mathematically unconstrained by the traffic demand. This degree of freedom physically corresponds to routing loops or circulatory flows, which consume network capacity without contributing to the net delivery of information.
	
	\subsection{Helmholtz-Hodge Decomposition of Traffic Fields}\label{S2.3}
	
	To quantify routing efficiency, the traffic flux density field is analyzed using the HHD \cite{Griffiths1999}, which decouples the flow into two orthogonal components with distinct physical interpretations as summarized in the following theorem.
	
	\begin{Theorem}\label{thm:traffic_decomp}
		For a macroscopic wireless traffic field $\mathbf{\Phi}(\mathbf{x})$ defined under the continuum limit (Assumption \ref{asm:continuum_limit}), there exists a unique orthogonal decomposition that partitions the network flow into a demand-driven irrotational component $\mathbf{\Phi}_{\mathrm{irr}}$ and a loop-induced solenoidal component $\mathbf{\Phi}_{\mathrm{sol}}$:
		\begin{equation} 
			\mathbf{\Phi}(\mathbf{x}) = \underbrace{-\nabla \psi(\mathbf{x})}_{\text{Potential Flow}} + \underbrace{\nabla \times \mathbf{A}(\mathbf{x})}_{\text{Rotational Flow}},
			\label{eq:traffic_mapping}
		\end{equation}
		where $\psi(\mathbf{x})$ is the scalar potential determined by the source-sink distribution $S(\mathbf{x})$, and $\mathbf{A}(\mathbf{x})$ is the vector potential characterizing routing vorticity. {\color{black}
			Thus, $\mathbf{\Phi}_{\mathrm{irr}}$ denotes the useful source-to-sink transport component, whereas $\mathbf{\Phi}_{\mathrm{sol}}$ denotes the loop-induced circulation component to be dissipated.
		}
	\end{Theorem}
	
	{\color{black}The main networking interpretation is as follows: the scalar potential $\psi$ describes the demand-driven descent direction from traffic sources to sinks, whereas the vector potential $\mathbf{A}$ accounts for circulation that does not change the net source-sink balance. Thus, the decomposition in \eqref{eq:traffic_mapping} is not an additional routing heuristic; it separates useful source-to-sink transport from loop-induced redundant circulation before optimization.}
	
	\begin{remark} \label{rem:boundary_logic} 
		In wireless networks, the boundary condition $\mathbf{\Phi} \cdot \mathbf{n} = 0$ on $\partial \Omega$, where $\mathbf{n}$ denotes the outward unit normal vector, represents a closed system where data packets do not exit the service area except through designated sinks. This constraint ensures the uniqueness of the decomposition in \eqref{eq:traffic_mapping}.
	\end{remark}
	
	The components in \eqref{eq:traffic_mapping} are derived from the potentials $\psi(\mathbf{x})$ and $\mathbf{A}(\mathbf{x})$. For a two-dimensional domain $\Omega \subset \mathbb{R}^2$, the vector potential is orthogonal to the plane, expressible as $\mathbf{A}(\mathbf{x}) = A_z(\mathbf{x})\mathbf{e}_z$, where $\mathbf{e}_z$ is the unit vector along the $z$-axis.
	Based on this formulation, the decomposed components yield the following specific physical interpretations for wireless networking.
	
	\subsubsection{Effective Flux ($\mathbf{\Phi}_{\mathrm{irr}}$)} 
	This component is curl-free, satisfying $\nabla \times \mathbf{\Phi}_{\mathrm{irr}} \equiv \mathbf{0}$. Substituting \eqref{eq:traffic_mapping} into the divergence constraint \eqref{eq:poisson_traffic} and utilizing the identity $\nabla \cdot (\nabla \times \mathbf{A}) \equiv 0$, the scalar potential is shown to satisfy the Poisson equation:
	\begin{equation} 
		-\nabla^2 \psi(\mathbf{x}) = S(\mathbf{x}).\label{eq:scalar_potential_poisson}
	\end{equation}
	The existence of a solution to this Neumann boundary value problem is mathematically guaranteed by the solvability condition established in Assumption \ref{asm:solvability}. The solution $\psi(\mathbf{x})$ is unique up to an additive constant, which is fixed by setting the spatial mean to zero ($\int_{\Omega} \psi \,\mathrm{d}\mathbf{x} = 0$) to ensure numerical stability.
	
	Since the potential field $\psi$ is mathematically unique (up to an additive constant), this relation indicates that the irrotational component $\mathbf{\Phi}_{\mathrm{irr}}$ is strictly determined by the spatial distribution of traffic sources and sinks, representing the fundamental transport path that satisfies network demand conservation.
	
	\subsubsection{Vortex Flux ($\mathbf{\Phi}_{\mathrm{sol}}$)} 
	This component is divergence-free, satisfying $\nabla \cdot \mathbf{\Phi}_{\mathrm{sol}} = 0$. In a closed domain, a divergence-free field corresponds to circulatory patterns that neither originate from nor terminate at sources or sinks. Consequently, $\mathbf{\Phi}_{\mathrm{sol}}$ represents routing loops and redundant detours that consume bandwidth without contributing to information delivery.
	
	The decomposition guarantees that $\mathbf{\Phi}_{\mathrm{irr}}$ and $\mathbf{\Phi}_{\mathrm{sol}}$ are orthogonal in the $L^2$ space, enabling the independent analysis of routing efficiency and traffic load.
	
	\subsection{Macroscopic Performance Metrics: Network Vorticity and Enstrophy}\label{S2.4}
	
	To characterize the macroscopic routing state, field-based metrics are introduced. Unlike traditional discrete metrics such as hop count or expected transmission count, these metrics evaluate the rotational properties of the continuous traffic field.
	
	\begin{Definition}\label{def:network_vorticity}
		The network vorticity field $\boldsymbol{\omega}(\mathbf{x})$ is defined as the curl of the traffic flux density field $\mathbf{\Phi}(\mathbf{x})$:
		\begin{equation} 
			\boldsymbol{\omega}(\mathbf{x}) \triangleq \nabla \times \mathbf{\Phi}(\mathbf{x}) = \left( \frac{\partial \Phi_y}{\partial x} - \frac{\partial \Phi_x}{\partial y} \right) \mathbf{e}_z,
			\label{eq:vorticity_def}
		\end{equation}
		where $\Phi_x$ and $\Phi_y$ denote the Cartesian components of the flux vector $\mathbf{\Phi}(\mathbf{x})$.
		{\color{black}In networking terms, a nonzero vorticity value indicates that local forwarding decisions contain a rotational component: packets tend to circulate among neighboring nodes instead of making net progress toward sinks. Therefore, vorticity provides a continuum-level measure of local loop tendency rather than merely a mathematical curl operator.}
	\end{Definition}
	
	\begin{Assumption}\label{asm:boundary_conditions}
		The network boundary $\partial \Omega$ is modeled as a rigid, frictionless reflection wall, characterized by the \textit{Free-Slip Boundary Condition}. This implies two physical constraints:
		\begin{enumerate}
			\item \textbf{No-Penetration (Kinematic):} Traffic flux cannot penetrate the boundary, i.e., $\mathbf{\Phi} \cdot \mathbf{n} = 0$, where $\mathbf{n}$ is the outward normal. This ensures mass conservation within $\Omega$.
			\item \textbf{Vanishing Vorticity (Dynamic):} The boundary does not exert viscous shear stress on the packet flow. Consequently, the vorticity at the boundary vanishes:
			\begin{equation} 
				\boldsymbol{\omega}(\mathbf{x}) = \nabla \times \mathbf{\Phi}(\mathbf{x}) = \mathbf{0}, \quad \forall \mathbf{x} \in \partial \Omega.
				\label{eq:boundary_omega}
			\end{equation}
		\end{enumerate}
		Mathematically, this set of conditions ensures the self-adjointness of the curl operator, guaranteeing the orthogonality of the HHD and the validity of the variational principle.
	\end{Assumption}

{\color{black}
	The free-slip condition is used as a sufficient reference boundary condition for deriving the clean Lyapunov dissipation result, rather than as a strict prerequisite for applying the HHD. For non-convex or multiply connected deployment regions, the decomposition can be generalized as
	\begin{equation}
		\mathbf{\Phi}
		= -\nabla \psi + \nabla \times \mathbf{A}
		+ \mathbf{\Phi}_{\mathrm{har}},
		\label{eq:generalized_hhd}
	\end{equation}
	where $\mathbf{\Phi}_{\mathrm{har}}$ is the harmonic component satisfying $\nabla\cdot\mathbf{\Phi}_{\mathrm{har}}=0$ and $\nabla\times\mathbf{\Phi}_{\mathrm{har}}=0$. Under the no-penetration condition and Hodge-compatible boundary conditions, the irrotational, solenoidal, and harmonic components remain $L^2$-orthogonal. Therefore, boundary-induced circulation or hole-bypass flow does not invalidate the HHD separation; instead, topology-induced harmonic flow is separated from loop-induced solenoidal vorticity.
}

	A non-zero vorticity $\boldsymbol{\omega}(\mathbf{x}) \neq \mathbf{0}$ signifies local traffic circulation, which physically corresponds to routing loops or inefficient detours. To evaluate the global performance of the network, the total network enstrophy $\mathcal{E}_{\mathrm{vor}}$ is defined as half the integral of the squared vorticity magnitude over the domain $\Omega$, consistent with standard fluid-kinetic conventions:
	\begin{equation} 
		\mathcal{E}_{\mathrm{vor}} \triangleq \frac{1}{2} \int_{\Omega} \|\boldsymbol{\omega}(\mathbf{x})\|^2 \, \mathrm{d}\mathbf{x}.
		\label{eq:enstrophy}
	\end{equation}
	
	The enstrophy $\mathcal{E}_{\mathrm{vor}}$ provides a scalar measure of the aggregate rotational intensity within the network. In an ideal routing configuration, the flow field is purely irrotational, leading to $\mathcal{E}_{\mathrm{vor}} \to 0$. Conversely, increasing values of $\mathcal{E}_{\mathrm{vor}}$ reflect higher resource consumption due to circulatory flows. {\color{black}Operationally, $\mathcal{E}_{\mathrm{vor}}$ aggregates the squared strength of local loop-inducing circulation over the entire service region. Thus, minimizing enstrophy means reducing the total intensity of routing-loop tendencies, rather than simply shortening paths or reducing traffic volume.} Minimizing $\mathcal{E}_{\mathrm{vor}}$ is therefore equivalent to aligning the traffic flow with the potential gradient, thereby reaching a potential flow regime that maximizes transport efficiency. Consequently, the routing optimization problem is formulated as the minimization of the enstrophy subject to network connectivity constraints.
	
	\section{Analytical Resolution and Energy Orthogonality}\label{S3}
	
	This section presents the derivation of the potentials obtained from the HHD and characterizes the energy separation between effective transport and routing redundancy. Given the observed traffic flux $\mathbf{\Phi}$, the objective is to determine the scalar potential $\psi$ and the vector potential $\mathbf{A}$ that uniquely satisfy the governing equations established in Section \ref{S2}.
	
	\subsection{Derivation of Traffic Potentials via Poisson Equations}\label{S3.1}
	
	Given the orthogonality of the HHD, the scalar and vector potentials are governed by independent Poisson equations subject to the domain's boundary conditions.
	
	\subsubsection{Scalar Potential (Neumann Problem)}\label{subsubsec:scalar_potential_res} 
	The irrotational component is driven by the source-sink distribution $S(\mathbf{x})$.
	Substituting $\mathbf{\Phi}_{\mathrm{irr}} = -\nabla \psi$ into the continuity equation $\nabla \cdot \mathbf{\Phi} = S$ yields the Poisson equation:
	\begin{equation} 
		\nabla^2 \psi(\mathbf{x}) = -S(\mathbf{x}), \quad \forall \mathbf{x} \in \Omega.
		\label{eq:scalar_poisson}
	\end{equation}
	The physical constraint that no traffic flows across the network boundary ($\mathbf{\Phi} \cdot \mathbf{n} = 0$ on $\partial \Omega$) implies a \textit{homogeneous Neumann boundary condition} for the scalar potential:
	\begin{equation} 
		\frac{\partial \psi}{\partial n} = -\mathbf{\Phi}_{\mathrm{irr}} \cdot \mathbf{n} = 0, \quad \text{on } \partial \Omega.
	\end{equation}
This boundary value problem is well-posed if and only if the global compatibility condition (Assumption \ref{asm:solvability}) is met, i.e., $\int_{\Omega} S(\mathbf{x}) \, \mathrm{d}\mathbf{x} = 0$. The solution $\psi(\mathbf{x})$ is unique up to an additive constant, which is typically fixed by setting the spatial mean to zero, namely, $\int_{\Omega} \psi \, \mathrm{d}\mathbf{x} = 0$. In practice, instead of explicitly evaluating the computationally expensive Green's function integral expressions, this Poisson equation under Neumann boundary conditions can be solved efficiently using standard numerical Poisson solvers, e.g., via the discrete cosine transform or iterative graph Laplacian methods, as detailed in Section \ref{S4.3}.
	
	\subsubsection{Vector Potential (Dirichlet Problem)}\label{subsubsec:vector_potential_res} 
	The solenoidal component is governed by the vorticity field $\boldsymbol{\omega} = \omega_z \mathbf{e}_z$.
	The vector potential $\mathbf{A} = A_z \mathbf{e}_z$ satisfies:
	\begin{equation} 
		\nabla^2 A_z(\mathbf{x}) = -\omega_z(\mathbf{x}).
		\label{eq:scalar_az_poisson}
	\end{equation}
	For the flow to be tangent to the boundary (implying no flux crossing $\partial \Omega$), the stream function $A_z$ must be constant along the boundary.
	Without loss of generality, we impose the \textit{homogeneous Dirichlet boundary condition}:
	\begin{equation} 
		A_z(\mathbf{x}) = 0, \quad \text{on } \partial \Omega.
	\end{equation}
	Similar to the scalar potential, the vector potential $A_z(\mathbf{x}) \mathbf{e}_z$ under Dirichlet boundary conditions can be readily obtained via standard grid-based numerical Poisson solvers, such as fast Fourier transform (FFT), tailored for bounded domains.
	
	\subsection{Orthogonal Decomposition of Transport Energy}\label{S3.2}
	
We quantify the macroscopic transport cost using the total transport energy functional $\mathcal{E}_{\mathrm{tot}}$, defined as the squared $L^2$-norm of the flux field over the network domain $\Omega$:
\begin{equation}
	\mathcal{E}_{\mathrm{tot}} \triangleq \int_{\Omega} \|\mathbf{\Phi}(\mathbf{x})\|^2 \,\mathrm{d}\mathbf{x}.
	\label{eq:total_energy}
\end{equation}
{\color{black}
	Here, $\mathcal{E}_{\mathrm{tot}}$ is a normalized transport-cost functional rather than a complete battery-energy model.
	At the discrete link level, radio resource consumption depends on link distance, channel condition, retransmission probability, transmission power control, and congestion-dependent resource occupation. If $f_{ij}$ denotes the packet forwarding load over link $(i,j)$ and $e_{ij}$ denotes the corresponding per-packet radio cost, a discrete radio-cost model can be written as
		$
		\mathcal{E}_{\mathrm{radio}}
		=
		\sum_{(i,j)\in\mathcal{E}}
		e_{ij} f_{ij}.
		$
		The continuum functional $\mathcal{E}_{\mathrm{tot}}$ should therefore be interpreted as a normalized surrogate of the aggregate forwarding load under spatial averaging and normalized link costs, rather than as a direct measurement of battery expenditure.
		}
	Based on the HHD established in Theorem~\ref{thm:traffic_decomp}, this energy functional can be decoupled into two orthogonal components, representing the effective transport effort and the routing redundancy, respectively.

	\begin{Proposition}\label{prop:ortho} 
		Consider a bounded domain $\Omega$ with boundary $\partial \Omega$. If the solenoidal component of the traffic flux satisfies the no-penetration boundary condition $\mathbf{\Phi}_{\mathrm{sol}} \cdot \mathbf{n} = 0$ on $\partial \Omega$, where $\mathbf{n}$ denotes the outward unit normal vector, the total transport energy is the sum of the irrotational energy and the solenoidal energy:
		\begin{equation} 
			\mathcal{E}_{\mathrm{tot}} = \mathcal{E}_{\mathrm{irr}} + \mathcal{E}_{\mathrm{sol}},
			\label{eq:energy_split}
		\end{equation}
		in which the component energies are defined as $\mathcal{E}_{\mathrm{irr}} = \int_{\Omega} \|\nabla \psi\|^2 \,\mathrm{d}\mathbf{x}$ and $\mathcal{E}_{\mathrm{sol}} = \int_{\Omega} \|\nabla \times \mathbf{A}\|^2 \,\mathrm{d}\mathbf{x}$, respectively.
	\end{Proposition}
	
	\begin{proof}
		The total transport energy is defined by the functional $\mathcal{E}_{\mathrm{tot}} \!=\! \int_{\Omega} \! \|\mathbf{\Phi}\|^2 \, \mathrm{d}\mathbf{x}$. Substituting the decomposition $\mathbf{\Phi}\! =\! -\nabla \!\psi \!+\! \mathbf{\Phi}_{\mathrm{sol}}$, the quadratic form of the integrand expands as:
		\begin{equation} 
			\|\mathbf{\Phi}\|^2 \!\! =\! \|\!-\! \nabla \psi \!+\! \mathbf{\Phi}_{\mathrm{sol}}\|^2 \!=\! \|\nabla \psi\|^2 \!+\! \|\mathbf{\Phi}_{\mathrm{sol}}\|^2 \!-\! 2 \nabla \psi \cdot \!\mathbf{\Phi}_{\mathrm{sol}}.
		\end{equation}
		Integrating over the domain $\Omega$, the total energy becomes:
		\begin{equation} 
			\mathcal{E}_{\mathrm{tot}} = \mathcal{E}_{\mathrm{irr}} + \mathcal{E}_{\mathrm{sol}} - 2 \int_{\Omega} \nabla \psi \cdot \mathbf{\Phi}_{\mathrm{sol}} \, \mathrm{d}\mathbf{x}.
		\end{equation}
		The proof of orthogonality hinges on the vanishing of the cross-term integral $I_{\mathrm{cross}} = \int_{\Omega} \nabla \psi \cdot \mathbf{\Phi}_{\mathrm{sol}} \, \mathrm{d}\mathbf{x}$. Using the vector calculus identity $\nabla \cdot (\psi \mathbf{\Phi}_{\mathrm{sol}}) = \nabla \psi \cdot \mathbf{\Phi}_{\mathrm{sol}} + \psi (\nabla \cdot \mathbf{\Phi}_{\mathrm{sol}})$, the integrand can be rewritten as:
		\begin{equation}\label{eqApA1}  
			\nabla \psi \cdot \mathbf{\Phi}_{\mathrm{sol}} = \nabla \cdot (\psi \mathbf{\Phi}_{\mathrm{sol}}) - \psi (\nabla \cdot \mathbf{\Phi}_{\mathrm{sol}}).
		\end{equation}
		By definition, the solenoidal component is divergence-free, i.e., $\nabla \cdot \mathbf{\Phi}_{\mathrm{sol}} \equiv 0$. Consequently, the second term on the right-hand side of (\ref{eqApA1}) vanishes, simplifying the integral to:
		\begin{equation} 
			I_{\mathrm{cross}} = \int_{\Omega} \nabla \cdot (\psi \mathbf{\Phi}_{\mathrm{sol}}) \, \mathrm{d}\mathbf{x}.
		\end{equation}
		Applying the Divergence Theorem, the volume integral transforms into a boundary flux integral:
		\begin{equation} 
			I_{\mathrm{cross}} = \oint_{\partial \Omega} \psi (\mathbf{\Phi}_{\mathrm{sol}} \cdot \mathbf{n}) \, \mathrm{d}l.
		\end{equation}
		Recall from Section \ref{subsubsec:vector_potential_res} that the vector potential satisfies the homogeneous Dirichlet boundary condition ($A_z|_{\partial \Omega} = 0$), which physically enforces that the boundary is a streamline of the solenoidal flow. 
		Mathematically, this implies $\mathbf{\Phi}_{\mathrm{sol}} \cdot \mathbf{n} = (\nabla \times A_z \mathbf{e}_z) \cdot \mathbf{n} = \frac{\partial A_z}{\partial \tau} = 0$, where $\tau$ is the tangential direction.
		Consequently, the boundary integral vanishes identically ($I_{\mathrm{cross}} = 0$), proving that the irrotational and solenoidal energies are strictly decoupled in the bounded domain.
	\end{proof}
	
	The derived orthogonality \eqref{eq:energy_split} confirms that the energy associated with routing loops is mathematically independent of the energy required for effective transport. {\color{black}From a routing perspective, this means that reducing the vortex component does not remove the source-to-sink potential component required for delivery. This is the reason VDR can suppress loop-induced redundancy without intentionally throttling useful traffic.} This theoretical separation allows us to define the vortex energy ratio (VER) as a normalized metric to quantify topological inefficiency:
	\begin{equation} 
		\eta_{\mathrm{vor}} \triangleq \frac{\mathcal{E}_{\mathrm{sol}}}{\mathcal{E}_{\mathrm{tot}}} = \frac{\int_{\Omega} \|\nabla \times \mathbf{A}\|^2 \,\mathrm{d}\mathbf{x}}{\int_{\Omega} \|\mathbf{\Phi}\|^2 \,\mathrm{d}\mathbf{x}}.
		\label{eq:ver_def}
	\end{equation}
	
	The metric $\eta_{\mathrm{vor}}$ characterizes the macroscopic flow regime of the network:
	\begin{itemize}
		\item \textit{Potential Flow Regime ($\eta_{\mathrm{vor}} \to 0$):} Represents an ideal transport state. Packet trajectories align strictly with the potential gradient $-\nabla \psi$, ensuring loop-free delivery.
		\item \textit{Vortex-Dominated Regime ($\eta_{\mathrm{vor}} \to 1$):} Indicates severe topological redundancy. A dominant fraction of traffic flux is trapped in circulatory structures, occupying capacity without contributing to net displacement.
	\end{itemize}
	
	{\color{black}
		The solenoidal component has a more specific operational meaning at the discrete network level. A nonzero solenoidal component corresponds to cyclic packet forwarding, where packets traverse additional links without changing the net source-sink balance. If $f^{\mathrm{sol}}_{ij}$ denotes the loop-induced component of the forwarding load on link $(i,j)$ and $e_{ij}$ denotes the corresponding per-packet radio cost, the redundant radio-resource cost can be represented as
		$
		\mathcal{E}_{\mathrm{red}}
		=
		\sum_{(i,j)\in\mathcal{E}}
		e_{ij} f^{\mathrm{sol}}_{ij}.
		$
		Reducing $\mathcal{E}_{\mathrm{sol}}$ therefore suppresses loop-induced redundant transmissions, which lowers channel occupation, retransmission exposure, and queueing pressure. This should be interpreted as reducing the redundant part of radio-resource consumption, rather than as solving a full physical power-control or battery-minimization problem.
	}
	\section{Dynamic Evolution and Vorticity Dissipation Control}\label{S4}
	
	While Section \ref{S3} establishes the static decomposition of the traffic field, robust network optimization necessitates a dynamic framework to actively suppress routing loops. In this section, we formulate the routing optimization problem as a continuous time-evolution process. By constructing a gradient flow on the network enstrophy functional, we derive a vorticity dissipation equation that guarantees the convergence of the traffic flux to a loop-free irrotational equilibrium.
	
	\subsection{Enstrophy Functional and VDR Formulation} \label{S4.1}
	
	To eliminate the solenoidal component responsible for routing loops, we adopt the network enstrophy as the objective functional for minimization.
	Unlike total energy minimization, which penalizes path length, enstrophy minimization explicitly targets the topological complexity (curl) of the flow, thereby inducing field smoothness and stability.
	The objective functional $\mathcal{F}[\mathbf{\Phi}]$ is exactly the network enstrophy defined in \eqref{eq:enstrophy}:
	\begin{equation} 
		\mathcal{F}[\mathbf{\Phi}] \triangleq \mathcal{E}_{\mathrm{vor}} = \frac{1}{2} \int_{\Omega} \|\nabla \times \mathbf{\Phi}(\mathbf{x})\|^2 \, \mathrm{d}\mathbf{x}.
		\label{eq:enstrophy_functional}
	\end{equation}
	{\color{black}
		For notation economy, $\mathcal{F}[\mathbf{\Phi}]$ is used only to emphasize the optimization objective, and it is physically identical to the network enstrophy $\mathcal{E}_{\mathrm{vor}}$.
	} We formulate  the VDR as a dynamic evolution process governed by the gradient flow of this functional. A virtual optimization time variable $\tau$ is introduced to characterize the iterative update. The evolution of the flux field $\mathbf{\Phi}(\mathbf{x}, \tau)$ follows the steepest descent direction:
	\begin{equation} 
		\frac{\partial \mathbf{\Phi}}{\partial \tau} = - \mu \frac{\delta \mathcal{F}}{\delta \mathbf{\Phi}},
		\label{eq:gradient_flow_def}
	\end{equation}
	where $\mu > 0$ denotes the mobility parameter. To determine the explicit form of the evolution equation, we first derive the functional gradient of the enstrophy.
	
	\begin{Lemma}\label{lem:frechet_derivative}
		Let $\delta \mathbf{\Phi}$ be an arbitrary infinitesimal perturbation satisfying the boundary condition $\delta \mathbf{\Phi} \times \mathbf{n} = \mathbf{0}$ on $\partial \Omega$. The Fréchet derivative (functional gradient) of the enstrophy functional $\mathcal{F}[\mathbf{\Phi}]$ is given by the double curl of the flux field:
		\begin{equation} 
			\frac{\delta \mathcal{F}}{\delta \mathbf{\Phi}} = \nabla \times (\nabla \times \mathbf{\Phi}).
		\end{equation}
	\end{Lemma}
	
	\begin{proof}
		We define the first variation of the enstrophy functional $\mathcal{F}$ with respect to an infinitesimal perturbation $\delta \mathbf{\Phi}$ as:
		\begin{equation} 
			\delta \mathcal{F} = \int_{\Omega} (\nabla \times \mathbf{\Phi}) \cdot (\nabla \times \delta \mathbf{\Phi}) \, \mathrm{d}\mathbf{x}.
		\end{equation}
		To isolate the perturbation term $\delta \mathbf{\Phi}$, we employ the vector identity regarding the divergence of a cross product:
		\begin{equation} 
			\nabla \cdot (\mathbf{A} \times \mathbf{B}) = \mathbf{B} \cdot (\nabla \times \mathbf{A}) - \mathbf{A} \cdot (\nabla \times \mathbf{B}).
		\end{equation}
		Rearranging terms and setting $\mathbf{A} = \nabla \times \mathbf{\Phi}$ and $\mathbf{B} = \delta \mathbf{\Phi}$, the integrand in the variation formulation can be rewritten as:
		\begin{equation} 
			(\nabla \!\times \!\mathbf{\Phi}) \! \cdot  \! (\nabla \! \times \! \delta \mathbf{\Phi}) \! =\!  \delta \mathbf{\Phi} \cdot  \! (\nabla \!\times \!(\nabla \!\times \! \mathbf{\Phi})) - \nabla \!\cdot ((\nabla \! \times\! \mathbf{\Phi}) \times \delta \mathbf{\Phi}) .
		\end{equation}
		Substituting this expansion back into the integral yields:
		\begin{equation} 
			\delta \mathcal{F} \!= \!\! \int_{\Omega}\! \delta \mathbf{\Phi} \!\cdot \! (\nabla \times (\nabla \times \mathbf{\Phi})) \, \mathrm{d}\mathbf{x} -\!\! \int_{\Omega}\!\! \nabla \cdot ((\nabla \times \mathbf{\Phi}) \times \delta \mathbf{\Phi}) \, \mathrm{d}\mathbf{x}.
		\end{equation}
		Applying the divergence theorem \cite{Arfken2013}, the second term converts to a boundary integral:
		\begin{equation} 
			I_{\mathrm{boundary}} = \oint_{\partial \Omega}\! ((\nabla \times \mathbf{\Phi}) \times \delta \mathbf{\Phi}) \cdot \mathbf{n} \, \mathrm{d}l.
		\end{equation}
		Using the vector triple product identity $(\mathbf{A} \times \mathbf{B}) \cdot \mathbf{C} = (\mathbf{C} \times \mathbf{A}) \cdot \mathbf{B}$, and substituting $\nabla \times \mathbf{\Phi} = \boldsymbol{\omega}$, the integrand becomes:
		\begin{equation} 
			((\nabla \times \mathbf{\Phi}) \times \delta \mathbf{\Phi}) \cdot \mathbf{n} = (\mathbf{n} \times \boldsymbol{\omega}) \cdot \delta \mathbf{\Phi}.
		\end{equation}
		Under the \textit{Free-Slip} assumption (Assumption \ref{asm:boundary_conditions}), the vorticity vanishes on the boundary ($\boldsymbol{\omega}|_{\partial \Omega} = \mathbf{0}$). 
		Consequently, the boundary integral $I_{\mathrm{boundary}}$ vanishes identically, independent of the perturbation $\delta \mathbf{\Phi}$.
		Thus, the first variation simplifies to:
		\begin{equation} 
			\delta \mathcal{F} = \int_{\Omega} \delta \mathbf{\Phi} \cdot [\nabla \times (\nabla \times \mathbf{\Phi})] \, \mathrm{d}\mathbf{x}.
		\end{equation}
		This completes the proof.
	\end{proof}
	
	Substituting the result from Lemma \ref{lem:frechet_derivative} into \eqref{eq:gradient_flow_def}, we obtain the governing vorticity dissipation equation:
	\begin{equation} 
		\frac{\partial \mathbf{\Phi}(\mathbf{x}, \tau)}{\partial \tau} = -\mu \left( \nabla \times (\nabla \times \mathbf{\Phi}(\mathbf{x}, \tau)) \right).
		\label{eq:evolution_pde}
	\end{equation}
	
	{\color{black}Operationally, this gradient-flow update can be read as a relaxation rule for forwarding fields: whenever neighboring forwarding directions create circulation, the solenoidal component is progressively damped, while the source-sink potential component is preserved. Therefore, the update in \eqref{eq:evolution_pde} directly implements vorticity dissipation at the field level.}

	\begin{remark}
		\label{rem:heat_equation}
		Equation \eqref{eq:evolution_pde} is analytically equivalent to a vector diffusion process. Using the Laplacian identity $\nabla \times (\nabla \times \mathbf{\Phi}) = \nabla(\nabla \cdot \mathbf{\Phi}) - \nabla^2 \mathbf{\Phi}$ and noting that $\nabla \cdot \mathbf{\Phi}$ is time-invariant, the evolution simplifies to $\partial_\tau \mathbf{\Phi}_{\mathrm{sol}} = \mu \nabla^2 \mathbf{\Phi}_{\mathrm{sol}}$. This implies that the control mechanism triggers a diffusive dissipation of the vorticity, effectively annealing topological loops.
	\end{remark}
	
	\subsection{Vorticity Dissipation Theorem}\label{S4.2}
	
	To rigorously analyze the convergence properties, we first examine the dynamics of the vorticity field itself. By applying the curl operator to the flow evolution equation \eqref{eq:evolution_pde}, we obtain the governing equation for routing loops:
	\begin{equation} 
		\frac{\partial \boldsymbol{\omega}}{\partial \tau} = \nabla \times \frac{\partial \mathbf{\Phi}}{\partial \tau} = -\mu \nabla \times (\nabla \times \boldsymbol{\omega}).
	\end{equation}
	Recalling that the vorticity field $\boldsymbol{\omega}$ is inherently solenoidal, i.e., $\nabla \cdot \boldsymbol{\omega} \equiv 0$, the vector identity $\nabla \times (\nabla \times \boldsymbol{\omega}) = \nabla(\nabla \cdot \boldsymbol{\omega}) - \nabla^2 \boldsymbol{\omega}$ simplifies the dynamics to a vector heat equation:
	\begin{equation} 
		\frac{\partial \boldsymbol{\omega}(\mathbf{x}, \tau)}{\partial \tau} = \mu \nabla^2 \boldsymbol{\omega}(\mathbf{x}, \tau).
		\label{eq:vorticity_heat_eqn}
	\end{equation}
	
	Equation \eqref{eq:vorticity_heat_eqn} reveals the fundamental physical mechanism of the VDR framework: the routing vorticity undergoes a  diffusion process. Similar to how thermal energy diffuses from high-temperature regions, topological routing loops spontaneously spread out and dissipate towards the boundaries, driven by the Laplacian operator $\nabla^2$.
	
	Based on this diffusion mechanism, the following theorem establishes the Lyapunov stability of the system.
	
	\begin{Theorem}\label{thm:enstrophy_dissipation} 
		Governed by the evolution dynamics in \eqref{eq:evolution_pde} and subject to the \textit{Ideal Confining Boundary Conditions} (Assumption \ref{asm:boundary_conditions}) which impose vanishing marginal vorticity ($\boldsymbol{\omega} = \mathbf{0}$ on $\partial \Omega$), the total network enstrophy $\mathcal{E}_{\mathrm{vor}}(\tau)$ decays monotonically with respect to the optimization time $\tau$:
		\begin{equation} 
			\frac{\mathrm{d}}{\mathrm{d}\tau}\mathcal{E}_{\mathrm{vor}}(\tau) = -\mu \int_{\Omega} \|\nabla \boldsymbol{\omega}\|^2 \, \mathrm{d}\mathbf{x} \leq 0.
			\label{eq:theorem_inequality}
		\end{equation}
	\end{Theorem}
	
	\begin{proof}
		The proof follows directly from the vorticity diffusion dynamics established in \eqref{eq:vorticity_heat_eqn}.
		Differentiating the enstrophy functional $\mathcal{E}_{\mathrm{vor}} = \frac{1}{2} \int_{\Omega} \|\boldsymbol{\omega}\|^2 \, \mathrm{d}\mathbf{x}$ yields:
		\begin{equation} 
			\frac{\mathrm{d} \mathcal{E}_{\mathrm{vor}}}{\mathrm{d}\tau} = \int_{\Omega} \boldsymbol{\omega} \cdot \frac{\partial \boldsymbol{\omega}}{\partial \tau} \, \mathrm{d}\mathbf{x}.
		\end{equation}
		Substituting the heat equation $\partial_\tau \boldsymbol{\omega} = \mu \nabla^2 \boldsymbol{\omega}$ into the integral:
		\begin{equation} 
			\frac{\mathrm{d} \mathcal{E}_{\mathrm{vor}}}{\mathrm{d}\tau} = \mu \int_{\Omega} \boldsymbol{\omega} \cdot (\nabla^2 \boldsymbol{\omega}) \, \mathrm{d}\mathbf{x}.
		\end{equation}
		Applying Green's first identity to the heat equation term:
		\begin{equation} 
			\int_{\Omega} \boldsymbol{\omega} \cdot \nabla^2 \boldsymbol{\omega} \, \mathrm{d}\mathbf{x} = \oint_{\partial \Omega} \boldsymbol{\omega} \cdot \frac{\partial \boldsymbol{\omega}}{\partial n} \, \mathrm{d}s - \int_{\Omega} \|\nabla \boldsymbol{\omega}\|^2 \, \mathrm{d}\mathbf{x}.
		\end{equation}
		Invoking the Vanishing Vorticity condition from Assumption \ref{asm:boundary_conditions} ($\boldsymbol{\omega} = \mathbf{0}$ on $\partial \Omega$), the surface integral term vanishes.
		Therefore, the time derivative of the enstrophy becomes strictly non-positive:
		\begin{equation} 
			\frac{\mathrm{d} \mathcal{E}_{\mathrm{vor}}}{\mathrm{d}\tau} = -\mu \int_{\Omega} \|\nabla \boldsymbol{\omega}\|^2 \, \mathrm{d}\mathbf{x} \leq 0.
		\end{equation}
		This confirms the monotonic dissipation of routing loops, completing the proof.
	\end{proof}
	
	Theorem \ref{thm:enstrophy_dissipation} confirms that the system does not exhibit oscillatory behavior. As $\tau \to \infty$, the enstrophy converges to zero, implying that the traffic flux asymptotically reaches the unique irrotational state $\mathbf{\Phi}_{\mathrm{irr}}$, which represents the global energy minimum within the solenoidal subspace.

	\begin{remark}
		While Theorem~\ref{thm:enstrophy_dissipation} assumes static topologies, it can be analytically shown that in time-varying environments, the system exhibits input-to-state stability. This ensures that the network vorticity is confined within a bounded neighborhood of zero, provided that the topological drift is bounded.
		{\color{black}
			A similar bounded-input interpretation applies when the ideal boundary condition $\left.\boldsymbol{\omega}\right|_{\partial\Omega}=\mathbf{0}$ is not exactly satisfied. In this case, the enstrophy balance becomes
			\begin{equation}
				\frac{\mathrm{d}\mathcal{E}_{\mathrm{vor}}}{\mathrm{d}\tau}
				=
				\mu \oint_{\partial\Omega}
				\boldsymbol{\omega}\cdot\frac{\partial\boldsymbol{\omega}}{\partial n}\,\mathrm{d}s
				-
				\mu \int_{\Omega} \|\nabla\boldsymbol{\omega}\|^2\,\mathrm{d}x .
				\label{eq:boundary_vorticity_balance}
			\end{equation}
			Assumption~5 corresponds to the special case where the boundary term vanishes. More generally, monotonic dissipation still holds for dissipative boundaries satisfying
			$\oint_{\partial\Omega}\omega \frac{\partial \omega}{\partial n}\,\mathrm{d}s \le 0$. If an irregular geographical boundary or coverage hole injects bounded local vorticity, the boundary integral acts as a bounded input. Consequently, VDR no longer guarantees exact convergence to zero enstrophy, but confines the network vorticity to a bounded residual neighborhood determined by the magnitude of the topology drift and boundary-induced circulation.
		}
	\end{remark}

	\subsection{Discrete Field Evolution on Graphs} \label{S4.3}
	
	To implement the VDR framework in practical wireless networks, we project the continuous field dynamics onto a discrete graph topology $\mathcal{G}(\mathcal{N}, \mathcal{E})$, where $\mathcal{N}$ represents the set of nodes and $\mathcal{E}$ the set of directional links. 
	This projection maps the infinite-dimensional Hilbert space of the continuum fields to finite-dimensional vector spaces. 
	Specifically, the continuous scalar potential $\psi(\mathbf{x})$ maps to the nodal potential vector $\bm{\psi}\! =\! \big[\psi_1, \cdots , \psi_{|\mathcal{N}|}\big]^{\top}\! \in\! \mathbb{R}^{|\mathcal{N}|}$, and the flux field $\mathbf{\Phi}(\mathbf{x})$ maps to the edge flow vector $\boldsymbol{\phi} \in \mathbb{R}^{|\mathcal{E}|}$. 
	The differential operators derived in the continuum limit are replaced by their spectral graph theoretic counterparts, ensuring that the physical conservation laws—mass conservation and vorticity dissipation—are preserved in the discrete domain.
	
	{\color{black}
		\subsubsection{Finite-density consistency analysis}
		Although the continuum formulation is derived under an asymptotic density assumption, the implemented VDR protocol operates on a finite random geometric graph. We therefore add a finite-density consistency estimate to clarify how the reconstructed continuum flux differs from the graph-projected field used by the discrete solver.

		\begin{Lemma}
			\label{lem:finite_density_projection}
			Let $\Omega\subset\mathbb{R}^2$ be a bounded domain with area $|\Omega|$ measured in $\mathrm{m}^2$, and let $\mathbf{\Phi}\in C^2(\Omega)^2$ denote the ideal continuum traffic flux. Consider a finite random geometric graph $G_\rho=(V_\rho,E_\rho)$ generated with node density $\rho$ nodes/km$^2$. Define the characteristic inter-node spacing as
			\begin{equation}
				\ell_\rho=\frac{1000}{\sqrt{\rho}}\quad \mathrm{m}.
				\label{eq:mean_spacing_density}
			\end{equation}
			Let $h$ be the kernel bandwidth and let $\Delta_g$ be the grid spacing of the numerical field solver. Suppose that $h>\ell_\rho$, $R_c\geq 2\ell_\rho$, and that the graph samples satisfy the usual local fill-distance condition with fill distance of order $\ell_\rho$ up to logarithmic factors. Then, with probability at least $1-N^{-c_0}$ for some constant $c_0>0$, the reconstructed graph-projected flux field $\widehat{\mathbf{\Phi}}_{G,h}$ satisfies
			\begin{align}
				&\left\|\widehat{\mathbf{\Phi}}_{G,h}-\mathbf{\Phi}\right\|_{L^2(\Omega)} \nonumber\\
				& \quad \leq C_1\Delta_g^2+C_2h^2+C_3\frac{\ell_\rho}{h}+C_4\sqrt{\frac{|\Omega|\log N}{Nh^2}},
				\label{eq:finite_density_flux_error}
			\end{align}
			where $N=|V_\rho|$. The constants $C_1$--$C_4$ depend on the domain regularity, the kernel shape, and the $C^2$ norm of $\mathbf{\Phi}$, but not on $N$, $\rho$, $\Delta_g$, or $h$. Moreover, the induced vorticity reconstruction error satisfies
			\begin{align}
				&\left\|\nabla\times\widehat{\mathbf{\Phi}}_{G,h}-\nabla\times\mathbf{\Phi}\right\|_{L^2(\Omega)} \nonumber\\
				&\quad  \leq C_5\Delta_g+C_6h+C_7\frac{\ell_\rho}{h^2}+C_8\sqrt{\frac{|\Omega|\log N}{Nh^4}}.
				\label{eq:finite_density_vorticity_error}
			\end{align}
		\end{Lemma}
		
		\begin{proof}
			We decompose the reconstruction error as
			\begin{align}
				\widehat{\mathbf{\Phi}}_{G,h}-\mathbf{\Phi}
				=&\left(\widehat{\mathbf{\Phi}}_{G,h}-\mathbf{\Phi}_{\rho,h}\right)+\left(\mathbf{\Phi}_{\rho,h}-\mathbf{\Phi}_{h}\right)\nonumber\\
				&+\left(\mathbf{\Phi}_{h}-\mathbf{\Phi}\right) \!+\! \left(\mathbf{\Phi}_{h}\!-\!\mathbf{\Phi}_{h}^{\Delta_g}\right)\!\!,
				\label{eq:error_decomposition}
			\end{align}
			where $\mathbf{\Phi}_{h}=K_h\ast\mathbf{\Phi}$ is the kernel-smoothed continuum field, $\mathbf{\Phi}_{\rho,h}$ is its finite-sample approximation, and $\mathbf{\Phi}_{h}^{\Delta_g}$ is the grid-discretized representation. By the triangle inequality,
			\begin{equation}
				\left\|\widehat{\mathbf{\Phi}}_{G,h}-\mathbf{\Phi}\right\|_{L^2}
				\leq E_{\mathrm{proj}}+E_{\mathrm{sam}}+E_{\mathrm{ker}}+E_{\mathrm{grid}}.
			\end{equation}
			The second-order finite-difference solver gives
			\begin{equation}
				E_{\mathrm{grid}}=\left\|\mathbf{\Phi}_{h}-\mathbf{\Phi}_{h}^{\Delta_g}\right\|_{L^2}\leq C_1\Delta_g^2.
			\end{equation}
			Since $\mathbf{\Phi}\in C^2(\Omega)^2$ and $K_h$ is symmetric and normalized, a Taylor expansion gives
			\begin{equation}
				\mathbf{\Phi}_{h}(\mathbf{x})-\mathbf{\Phi}(\mathbf{x})
				=\int K_h(\mathbf{x}-\mathbf{y})\left[\mathbf{\Phi}(\mathbf{y})-\mathbf{\Phi}(\mathbf{x})\right]\,\mathrm{d}\mathbf{y}
				=O(h^2),
			\end{equation}
			where the first-order term vanishes by kernel symmetry. Hence
			\begin{equation}
				E_{\mathrm{ker}}=\left\|\mathbf{\Phi}_{h}-\mathbf{\Phi}\right\|_{L^2}\leq C_2h^2.
			\end{equation}
			For the graph projection term, the local fill-distance condition implies that the graph-supported quadrature resolves a kernel neighborhood of radius $O(h)$ up to a normalized coverage error of order $\ell_\rho/h$. Therefore,
			\begin{equation}
				E_{\mathrm{proj}}\leq C_3\frac{\ell_\rho}{h}.
			\end{equation}
			The random node locations introduce a finite-sample fluctuation. Since the expected number of samples in a kernel neighborhood scales as $Nh^2/|\Omega|$, standard concentration for two-dimensional kernel averages yields
			\begin{equation}
				E_{\mathrm{sam}}\leq C_4\sqrt{\frac{|\Omega|\log N}{Nh^2}}
			\end{equation}
			with probability at least $1-N^{-c_0}$. Combining the four estimates proves \eqref{eq:finite_density_flux_error}.
			
			For the vorticity error, the curl operator differentiates the reconstructed flux once. The grid term is reduced from $O(\Delta_g^2)$ to $O(\Delta_g)$, the kernel bias contributes $O(h)$ after differentiation, and the projection and sampling terms acquire one additional factor of $1/h$. This gives
			\begin{align}
				&	\left\|\nabla\times\widehat{\mathbf{\Phi}}_{G,h}-\nabla\times\mathbf{\Phi}\right\|_{L^2} \nonumber\\
				&\quad	\leq C_5\Delta_g+C_6h+C_7\frac{\ell_\rho}{h^2}+C_8\sqrt{\frac{|\Omega|\log N}{Nh^4}},
			\end{align}
			which proves \eqref{eq:finite_density_vorticity_error}.
		\end{proof}
		
		For the density range used in the simulations, the characteristic inter-node spacing is
		\[
		\ell_\rho\approx 31.6~\mathrm{m},~22.4~\mathrm{m},~18.3~\mathrm{m},~15.8~\mathrm{m}
		\]
		for $\rho=1000,2000,3000,4000$ nodes/km$^2$, respectively. With the default bandwidth $h=R_c/2=90$ m, the corresponding graph-projection ratios are
		$
		\ell_\rho/h\approx 0.351,~0.249,~0.204,~0.176.
		$
		Thus, the finite-density projection term decreases over the evaluated density range. This result does not turn the finite graph into an infinite-density continuum model; rather, it shows that the continuum-to-graph projection error is explicitly retained and becomes smaller as the simulated network becomes denser.
		
		This lemma clarifies the role of the continuum approximation in the finite-density implementation. After the flux field is projected onto the graph, the discrete solver in the following subsections operates directly on the finite graph. The lemma therefore provides a consistency bound for the projection step, whereas packet-level loop closure is handled separately by the forwarding rule in Section~\ref{S4.4}.

	}

	\subsubsection{Discrete Scalar Field Solver} 
	To numerically resolve the macroscopic potential field, we discretize the Poisson equation $-\nabla^2 \psi = S$ using the graph Laplacian matrix. 
	Let $\mathbf{W} \in \mathbb{R}^{|\mathcal{E}| \times |\mathcal{E}|}$ be the diagonal matrix of link weights (representing channel capacity or conductance) and $\mathbf{B} \in \mathbb{R}^{|\mathcal{N}| \times |\mathcal{E}|}$ be the node-edge incidence matrix. 
	The weighted graph Laplacian $\mathbf{L} \triangleq \mathbf{B} \mathbf{W} \mathbf{B}^{\top}$ serves as the discrete negative Laplacian operator. 
	A fundamental property of $\mathbf{L}$ is that its row sums are identically zero ($\mathbf{L}\mathbf{1} = \mathbf{0}$), which mathematically enforces the homogeneous Neumann boundary condition ($\partial \psi / \partial n = 0$) inherent to closed network systems.
	
	The discretization yields the linear system:
	\begin{equation} 
		\mathbf{L} \bm{\psi} = \mathbf{s},
		\label{eq:discrete_poisson}
	\end{equation}
	where $\mathbf{s}\! =\! \big[s_1,\cdots ,s_{|\mathcal{N}|} \big]^{\top}\! \in\! \mathbb{R}^{|\mathcal{N}|}$ is the source intensity vector. 
	Subject to the global solvability condition $\sum_{i} s_i\! =\! 0$ (traffic balance), this system admits a unique solution up to an additive constant. 
	In our distributed implementation, we solve this system via a Jacobi-based relaxation method. 
	At iteration $k$, node $i$ updates its potential $\psi_i$ using the weighted average of its neighbors $\mathcal{N}_i$:
	\begin{equation} 
		\psi_i^{(k+1)} = \frac{1}{\sum_{j \in \mathcal{N}_i} w_{ij}} \left( s_i + \sum_{j \in \mathcal{N}_i} w_{ij} \psi_j^{(k)} \right),
		\label{eq:jacobi_update}
	\end{equation}
	where $w_{ij}$ denotes the link weight or conductance between node $i$ and node $j$. Note that mathematically, if an edge $e$ connects nodes $i$ and $j$, then $w_{ij}$ corresponds to the diagonal entry of the diagonal matrix $\mathbf{W}$ indexed by $e$.
	This local interaction diffuses the source intensity $s_i$ across the network, asymptotically minimizing the discrete Dirichlet energy $\mathcal{E}_{\mathrm{pot}} = \frac{1}{2} \boldsymbol{\psi}^{\top} \mathbf{L} \boldsymbol{\psi} - \mathbf{s}^{\top} \boldsymbol{\psi}$.
	
	\paragraph*{Source Intensity Estimation}
	To align with the continuity equation $\nabla \cdot \mathbf{\Phi} = S - \partial \rho/\partial t$, the discrete source term $s_i$ is strictly defined as the \textit{net traffic generation rate} (mean arrival rate minus sink consumption), distinct from instantaneous queue variations $\partial Q_i/\partial t$. 
	This definition decouples global equilibrium planning from transient queue stabilization.
	
	Nodes acting as sources estimate their generation intensity $s_i>0$ from the packet arrival rate
	\begin{equation}
		a_i[t] \triangleq \frac{N_{\mathrm{pkt}}(t,t+\Delta t)}{\Delta t}.
	\end{equation}
	{\color{black}
		To improve robustness under bursty machine-type traffic, we replace the fixed EMA by a burst-aware EMA:
		\begin{equation}
			s_i[t]
			=
			\bigl(1-\gamma_i[t]\bigr)s_i[t-1]
			+
			\gamma_i[t]a_i[t],
			\label{eq:adaptive_source_ema}
		\end{equation}
		where
		\begin{align}
			&\gamma_i[t]
			=
			\gamma_{\min}
			+
			(\gamma_{\max}-\gamma_{\min})
			\frac{\chi_i[t]}{\chi_i[t]+\theta}, \nonumber\\
			&\chi_i[t]
			=
			\frac{|a_i[t]-s_i[t-1]|}{s_i[t-1]+\epsilon_{\rm ema}}.
			\label{eq:adaptive_gamma}
		\end{align}
		Here, $\gamma_{\min}=0.1$, $\gamma_{\max}=0.5$, and $\theta=1$. Under stable traffic, $\gamma_i[t]$ remains close to $\gamma_{\min}$ and preserves smoothing; when a burst is detected, $\gamma_i[t]$ increases toward $\gamma_{\max}$ and shortens the source-field response time.
	}
	Relay nodes with no local traffic maintain $s_i \approx 0$. 
	To satisfy the Neumann solvability condition, the sink node adaptively sets its intensity $s_{\mathrm{sink}}$ to match the negative sum of network influxes, $s_{\mathrm{sink}} = - \sum_{k \neq \mathrm{sink}} s_k$, estimated via total received throughput.
	
	{\color{black}
		For a step change in the observed arrival intensity, the EMA tracking error decays as
		\begin{equation}
			|e_i[t+k]|
			\le
			(1-\gamma)^k |e_i[t]|.
		\end{equation}
		Thus, the $90\%$ response time is
		\begin{equation}
			k_{90}(\gamma)
			=
			\left\lceil
			\frac{\log(0.1)}{\log(1-\gamma)}
			\right\rceil.
		\end{equation}
		The fixed setting $\gamma=0.1$ gives $k_{90}\approx 22$ update windows, whereas the burst mode $\gamma_{\max}=0.5$ reduces this value to about $4$ update windows. This change affects only the control-plane source-intensity estimator; the packet-level forwarding utility remains unchanged.
	}
	
	\paragraph*{Asynchronous Event-Triggered Updates}
	In large-scale networks, continuous global synchronization incurs significant signaling overhead. 
	We implement an asynchronous mechanism where node $i$ broadcasts its updated potential $\psi_i^{(k+1)}$ only when the relative divergence exceeds a threshold $\delta$:
	\begin{equation} 
		\frac{|\psi_i^{(k+1)} - \psi_i^{(k)}|}{|\psi_i^{(k)}| + \epsilon} > \delta,
		\label{eq:update_trigger}
	\end{equation}
	where $\epsilon$ is a regularization constant. 
	This ensures that control signaling is sparse, generated primarily during topological perturbations or traffic bursts.
	
	\subsubsection{Discrete Flux Evolution and Gradient Projection} 
	While the continuous vorticity dissipation is governed by the double-curl operator (\ref{eq:evolution_pde}), directly discretizing this operator via the graph Hodge Laplacian requires evaluating circulation over all 2-simplices (triangles or cliques) in the network. 
	In a distributed wireless environment, tracking such higher-order topological structures incurs prohibitive signaling overhead.
	
	To achieve a practical distributed protocol, we employ a first-order gradient projection method that bypasses the explicit computation of discrete curls, while asymptotically reaching the exact same irrotational equilibrium. 
	For each directed link $e=(i,j)$, the discrete flux $\phi_{ij}$, i.e., the element of $\bm{\phi}$ indexed by $e$, at iteration $\tau$ is updated via a convex relaxation:
	\begin{equation} 
		\phi_{ij}^{(\tau+1)} = (1-\zeta) \phi_{ij}^{(\tau)} + \zeta w_{ij} \left( \psi_i - \psi_j \right),
		\label{eq:discrete_update}
	\end{equation}
	where $\zeta \in (0, 1]$ is the relaxation step size. In vector form, this update is expressed as $\boldsymbol{\phi}^{(\tau+1)} = (1-\zeta) \boldsymbol{\phi}^{(\tau)} + \zeta \mathbf{W} \mathbf{B}^{\top} \boldsymbol{\psi}$. 
	By iteratively mixing the current flow with the ideal potential-driven flow, Eq. \eqref{eq:discrete_update} acts as a low-pass topological filter. 
	The theoretical guarantee of this engineering approximation is established below.
	
	\begin{Proposition}\label{prop:discrete_convergence}
		Let $\boldsymbol{\phi}_{\mathrm{irr}}^* = \mathbf{W} \mathbf{B}^{\top} \boldsymbol{\psi}$ denote the ideal irrotational edge flow, and let the discrete traffic flux be orthogonally decomposed according to the discrete Helmholtz theorem as $\boldsymbol{\phi}^{(\tau)} = \boldsymbol{\phi}_{\mathrm{irr}}^* + \boldsymbol{\phi}_{\mathrm{sol}}^{(\tau)}$, where $\boldsymbol{\phi}_{\mathrm{sol}}^{(\tau)} \in \mathrm{Ker}(\mathbf{B})$ represents the discrete routing loops. 
		Under the update rule \eqref{eq:discrete_update} with symmetric link weights ($w_{ij} = w_{ji}$) and a step size $\zeta \in (0, 1]$, the discrete solenoidal component decays exponentially:
		\begin{equation} 
			\boldsymbol{\phi}_{\mathrm{sol}}^{(\tau)} = (1-\zeta)^{\tau} \boldsymbol{\phi}_{\mathrm{sol}}^{(0)}.
		\end{equation}
		Consequently, the discrete loop energy vanishes asymptotically ($\lim_{\tau \to \infty} \|\boldsymbol{\phi}_{\mathrm{sol}}^{(\tau)}\|^2 = 0$), mathematically mirroring the continuous Lyapunov stability (Theorem \ref{thm:enstrophy_dissipation}) without requiring explicit Hodge Laplacian computations.
	\end{Proposition}
	
	\begin{proof}
		Subtracting the steady-state irrotational flow $\boldsymbol{\phi}_{\mathrm{irr}}^*$ from both sides of the vector update equation $\boldsymbol{\phi}^{(\tau+1)} = (1-\zeta) \boldsymbol{\phi}^{(\tau)} + \zeta \boldsymbol{\phi}_{\mathrm{irr}}^*$ yields the error dynamics:
		\begin{equation} 
			\boldsymbol{\phi}^{(\tau+1)} - \boldsymbol{\phi}_{\mathrm{irr}}^* = (1-\zeta) (\boldsymbol{\phi}^{(\tau)} - \boldsymbol{\phi}_{\mathrm{irr}}^*).
		\end{equation}
		Since $\boldsymbol{\phi}^{(\tau)} - \boldsymbol{\phi}_{\mathrm{irr}}^* \triangleq \boldsymbol{\phi}_{\mathrm{sol}}^{(\tau)}$, this simplifies to $\boldsymbol{\phi}_{\mathrm{sol}}^{(\tau+1)} = (1-\zeta) \boldsymbol{\phi}_{\mathrm{sol}}^{(\tau)}$. 
		For any constant step size $\zeta \in (0, 1]$, the contraction mapping factor $|1-\zeta| < 1$, ensuring geometric convergence of the solenoidal error to zero.
	\end{proof}
	
	\begin{remark} 
		This discrete formulation ensures spectral consistency with the continuous theory. The update rule \eqref{eq:discrete_update} effectively projects the initial traffic distribution onto the gradient subspace $\mathrm{Im}(\mathbf{B}^{\top})$, which is orthogonal to the cycle subspace (the kernel of the discrete curl). Consequently, as $\tau \to \infty$, the discrete circulation vanishes, guaranteeing the asymptotic loop-free property established in Theorem \ref{thm:enstrophy_dissipation}.
	\end{remark}
	
	\subsection{Distributed Packet Forwarding Protocol}\label{S4.4}

	To bridge the timescale discrepancy between macroscopic field evolution and microscopic packet dynamics, the VDR framework operates on a rigorous closed-loop architecture. At the slow timescale (macro-level), the network periodically reconstructs the continuous demand fields and executes the discrete gradient projection (Section \ref{S4.3}) to dissipate topological routing loops. Concurrently, at the fast timescale (micro-level), individual nodes utilize these optimized macroscopic fields to make instantaneous, per-packet routing decisions. This multi-timescale separation ensures that transient congestion is handled locally, while the global transport remains asymptotically loop-free.
	
	{\color{black}
		The virtual optimization index $\tau$ used in the field evolution is not a physical packet-forwarding time slot. In practical operation, potential-field updates are performed at the control-plane refresh timescale and can be warm-started from the previous field, while packets continue to be forwarded using the latest available potentials. Thus, the relaxation iterations used to illustrate convergence are amortized over the control window and are not incurred as per-packet signalling latency.}
	
	{\color{black}
		If a control-plane refresh is executed, its overhead can be decomposed as
		$
		T_{\mathrm{ctrl}}
		=
		T_{\mathrm{state}}
		+
		K_{\tau}T_{\mathrm{relax}}
		+
		T_{\mathrm{dist}},
		$
		where $T_{\mathrm{state}}$ is the state-collection time, $T_{\mathrm{relax}}$ is the cost of one relaxation/update step, $K_{\tau}$ is the number of relaxation steps within the refresh window, and $T_{\mathrm{dist}}$ is the potential-distribution time. This overhead belongs to the control plane and is amortized over the refresh window; it is not added to each packet forwarding decision.
	}
	
	To execute these micro-level decisions, we implement a deterministic greedy forwarding policy. 
	To balance the trade-off between latency minimization (gradient drift) and throughput maximization (pressure diffusion), we introduce a hybrid forwarding utility metric. 
	{\color{black}To avoid mixing quantities with different units and scales, the potential and queue differences are normalized before aggregation. For $j\in\mathcal{N}_i$, let}
	\begin{equation}
		{\color{black}
			\widehat{\Delta\psi}_{ij}
			\!=\!
			\frac{\psi_i-\psi_j}
			{\max_{k\in\mathcal{N}_i}|\psi_i-\psi_k|+\varepsilon},
			\quad \!
			\widehat{\Delta Q}_{ij}
			\!=\!
			\frac{Q_i-Q_j}{Q_{\max}},
		}
		\label{eq:utility_normalization}
	\end{equation}
	{\color{black}
		where $\varepsilon>0$ avoids division by zero and $Q_{\max}$ is the queue-capacity normalization constant. Both terms are therefore dimensionless. Moreover, $|\widehat{\Delta\psi}_{ij}|\leq 1$ by construction, and $|\widehat{\Delta Q}_{ij}|\leq 1$ when $0\leq Q_i,Q_j\leq Q_{\max}$. Therefore, the two components entering the weighted utility are scale-aligned before aggregation.
	}
	For a packet $p$ at node $i$, the forwarding utility $U_{ij}$ for each candidate neighbor $j \in \mathcal{N}_i$ is defined as:
	\begin{equation}
		U_{ij} \triangleq 
		\underbrace{\alpha_i \widehat{\Delta\psi}_{ij}}_{\text{Normalized Potential Drive}} 
		+ 
		\underbrace{(1-\alpha_i) \widehat{\Delta Q}_{ij}}_{\text{Normalized Queue Pressure}} 
		- 
		\underbrace{\Lambda(j, \mathcal{H}_p)}_{\text{Loop Penalty}},
		\label{eq:utility_def}
	\end{equation}
	{\color{black}
		where $Q_i$ denotes the local queue length, the normalized terms are defined in \eqref{eq:utility_normalization}, and $\Lambda(j, \mathcal{H}_p)$ represents the packet-level loop-avoidance barrier:
	}
	\begin{equation}
		\Lambda(j, \mathcal{H}_p) =
		\begin{cases}
			\infty, & \text{if } j \in \mathcal{H}_p, \\
			0, & \text{otherwise}.
		\end{cases}
		\label{eq:penalty_func}
	\end{equation}
	Here, $\mathcal{H}_p$ denotes the packet trajectory history. 
	{\color{black}
		The penalty $\Lambda(j,\mathcal{H}_p)$ is not a physical queue or potential quantity to be balanced against the two normalized terms; it is a hard exclusion barrier used to remove candidates that would close a realized forwarding loop. The local weighting factor $\alpha_i(t) \in [0,1]$ therefore balances the normalized potential-drive term and the normalized queue-pressure term, giving it a consistent interpretation independent of the raw units of $\psi_i$ and $Q_i$
	}
	Instead of fixing a global parameter, we define $\alpha_i(t)$ as a monotonically decreasing function of the normalized local queue length $q_i(t) = Q_i(t) / {\color{black}Q_{\max}}$:
	\begin{equation} 
		\alpha_i(t) = \alpha_{\mathrm{max}} \cdot \left( 1 - q_i(t) \right)^k,
		\label{eq:adaptive_alpha}
	\end{equation}
	where $\alpha_{\mathrm{max}} \in (0, 1]$ is the maximum gradient sensitivity and $k \ge 1$ is a shaping factor. This formulation creates two distinct operating regimes. Under light loads ($q_i \to 0$), $\alpha_i$ approaches $\alpha_{\max}$, biasing the decision towards the potential gradient for delay optimality. Conversely, under heavy loads ($q_i \to 1$), $\alpha_i$ vanishes, transitioning the mechanism into pure backpressure routing to maximize throughput via load balancing.
	
	Finally, the next hop $j^*$ is selected to maximize the utility $U_{ij}$, as detailed in Algorithm \ref{alg:vdr_logic}, where the variable $U_{\max}$ tracks the optimal candidate and $T_{\max}$ enforces a standard time-to-live (TTL) constraint to prevent indefinite circulation.

	{\color{black}
		The zero-loop guarantee in Algorithm~\ref{alg:vdr_logic} refers to realized forwarding cycles, i.e., a packet is never forwarded to a node already recorded in its trajectory history. A Bloom-filter hit does not create a loop; it only blocks the corresponding candidate next hop. If all candidates are blocked or have non-positive utility, the packet is temporarily held in the local buffer and re-evaluated in the next scheduling slot using the updated queue states and potential field. This waiting time is included in the end-to-end delay measurement. If the packet exceeds the TTL or the local buffer overflows, it is counted as a packet drop and is therefore reflected in the PDR. Thus, loop prevention is not excluded from the performance statistics: blocked loop-closure attempts appear as temporary stalls or drops, whereas Fig.~\ref{fig:loop_perf} reports only realized closed loops.
	}
	
	\begin{algorithm}[!t]
		\footnotesize
		\caption{VDR Distributed Forwarding Protocol}
		\label{alg:vdr_logic}
		\begin{algorithmic}[1]
			\Require Current Node $i$; Packet $p$ with trajectory Bloom filter $\mathbf{B}_p$; Neighbor set $\mathcal{N}_i$; Potentials $\boldsymbol{\psi}$; Queues $\mathbf{Q}$.
			\Ensure Next Hop Node $j^*$.
			
			\State \textbf{Initialization:} $j^* \leftarrow i$, $U_{\max} \leftarrow 0$. \Comment{Strict positive-utility threshold}
			\State \textbf{Update Weight:} Calculate local sensitivity $\alpha_i$ via \eqref{eq:adaptive_alpha}.
			\State {\color{black}$\Psi_{\max} \leftarrow \max_{k\in\mathcal{N}_i}|\psi_i-\psi_k|+\varepsilon$}
			
			\State \textit{// Step 1: Check Termination Conditions}
			\If{$i$ is Sink node} 
			\State \Return \textbf{Packet Delivered} 
			\EndIf
			\If{$\mathrm{HopCount}(p) \ge T_{\max}$} 
			\State \Return \textbf{Packet Dropped} 
			\EndIf
			
			\State \textit{// Step 2: Evaluate Candidates via Hybrid Metric}
			\For{each neighbor $j \in \mathcal{N}_i$}
			\State {\color{black}$\widehat{\Delta\psi}_{ij} \leftarrow \frac{\psi_i-\psi_j}{\Psi_{\max}}$}
			\State {\color{black}$\widehat{\Delta Q}_{ij} \leftarrow \frac{Q_i-Q_j}{Q_{\max}}$}

			\If{$\text{Query}(\mathbf{B}_p, j) == \text{TRUE}$} \Comment{Loop detected in history}
			\State $\Lambda \leftarrow \infty$ 
			\Else
			\State $\Lambda \leftarrow 0$
			\EndIf
			
			\State {\color{black}$U_{ij} \leftarrow \alpha_i \cdot \widehat{\Delta\psi}_{ij} + (1-\alpha_i) \cdot \widehat{\Delta Q}_{ij} - \Lambda$}
			
			\If{$U_{ij} > U_{\max}$}
			\State $U_{\max} \leftarrow U_{ij}$; \quad $j^* \leftarrow j$
			\EndIf
			\EndFor
			
			\State \textit{// Step 3: Execution}
			\If{$j^* \neq i$} 
			\State Add $\text{ID}_i$ to $\mathbf{B}_p$ \Comment{Record self before departing}
			\State Forward packet $p$ to $j^*$ 
			\Else 
			\State {\color{black}Hold packet in local buffer and re-evaluate in the next slot} 
			\Comment{\color{black}Waiting time counted in delay; drop if TTL or buffer limit is exceeded}
			\EndIf
		\end{algorithmic}
	\end{algorithm}
	
	\begin{remark} 
		{\color{black}
			Lemma~\ref{lem:finite_density_projection} bounds the continuum-to-graph reconstruction and projection error, but it does not by itself prove packet-level loop freedom on a finite graph. In particular, finite grid and projection errors may perturb near-tie forwarding utilities. For this reason, finite-graph loop-closure prevention is enforced at the packet-forwarding layer by the trajectory history guard. Since the TTL is $64$ hops, at most $n_{\max}=64$ node identifiers are inserted into the Bloom filter. We set the Bloom-filter size to $m=640$ bits and use $k_{\rm BF}=7$ hash functions, giving
			\begin{equation}
				p_{\rm fp}=\left(1-e^{-k_{\rm BF}n_{\max}/m}\right)^{k_{\rm BF}}\approx 8.2\times10^{-3}.
				\label{eq:bloom_fp}
			\end{equation}
			A previously visited node is therefore not accepted as a valid next hop because standard Bloom filters have no false negatives. False positives may only block feasible candidates and trigger temporary holding or re-evaluation; they do not create a completed forwarding cycle. The per-packet memory overhead is $640$ bits ($80$ bytes), and the worst-case nodal overhead is $80Q_{\max}$ bytes, i.e., $3200$ bytes for $Q_{\max}=40$.
		}
	\end{remark}

	\subsection{Macroscopic Network Observability and Anomaly Diagnosis}\label{S4.5}
	
	While the VDR protocol (Sections \ref{S4.3}--\ref{S4.4}) inherently suppresses routing loops via vorticity dissipation, practical network management requires a mechanism to visually monitor and verify the topological health of the network.
	However, the direct application of continuous vector operators on discrete, irregular graph topologies introduces specific numerical artifacts.
	A primary challenge arises from high-density convergent flows near traffic sinks.
	Due to the discretization error of the gradient operator on non-uniform meshes, these strong radial fluxes often manifest as spurious vorticity (false positives), potentially masking low-intensity, genuine routing loops.
	
	To enable robust diagnosis in such noisy environments, we propose a physics-based filtering mechanism that exploits the distinct kinematic signatures of convergent versus circulatory fluxes.
	Consistent with the system model in Section \ref{S2.3}, convergent flows (sinks) are characterized by high negative flux divergence ($|\nabla \cdot \mathbf{\Phi}| \approx |S| \gg 0$) with low intrinsic vorticity, whereas circulatory flows (loops) exhibit near-zero divergence ($|\nabla \cdot \mathbf{\Phi}| \approx 0$) but high flux vorticity ($\|\nabla \times \mathbf{\Phi}\| > 0$).
	Leveraging this property, we define a divergence-based filter to generate a diagnostic heatmap, denoted as $\mathcal{H}(\mathbf{x})$, which visualizes the spatial distribution of genuine routing anomalies.
	The filtered vorticity magnitude $\omega_{\mathrm{filt}}(\mathbf{x})$ is computed by selectively suppressing regions where the local flux is divergence-dominated:
	\begin{equation} 
		\omega_{\mathrm{filt}}(\mathbf{x})\!\! = \!\!
		\begin{cases} 
			\hspace*{6mm}0, \!\!&\!\!\! \text{if } \!|\nabla \!\cdot  \! \mathbf{\Phi}(\mathbf{x})| \!> \!  \beta \|\nabla \! \times \! \mathbf{\Phi}(\mathbf{x})\|, \\
			\|\nabla \times \mathbf{\Phi}(\mathbf{x})\| , \!\! &\!\!\! \text{otherwise},
		\end{cases}
		\label{eq:div_filter}
	\end{equation}
	where $\beta > 0$ is a sensitivity threshold (typically $\beta \approx 1.0$).
	This filter effectively removes the numerical artifacts induced by strong sink convergence (where $\nabla \rho$ is large) while preserving the topological signature of genuine routing loops ($\boldsymbol{\omega} = \nabla \times \mathbf{\Phi}$).
	The complete diagnostic procedure, which transforms a set of raw discrete routing traces $\mathcal{T}$, representing the collected packet forwarding logs, into the anomaly heatmap $\mathcal{H}(\mathbf{x})$, is summarized in Algorithm \ref{alg:diagnosis}.
	
	\begin{algorithm}[!t]
		\footnotesize
		\caption{Macroscopic Anomaly Detection via Hodge Filtering}
		\label{alg:diagnosis}
		\begin{algorithmic}[1]
			\Require Discrete flow traces $\mathcal{T} = \{ (\mathbf{x}_k, \boldsymbol{\phi}_k) \}$; Smoothing bandwidth $h$; Threshold $\beta$.
			\Ensure Diagnostic Heatmap $\mathcal{H}(\mathbf{x})$.
			
			\State \textit{// Step 1: Field Reconstruction}
			\State Reconstruct continuous flux field $\mathbf{\Phi}(\mathbf{x})$ from $\mathcal{T}$ via vector KDE with bandwidth $h$.
			
			\State \textit{// Step 2: Differential Extraction}
			\State Compute flux vorticity: $\boldsymbol{\omega}(\mathbf{x}) \leftarrow \nabla \times \mathbf{\Phi}(\mathbf{x})$
			\State Compute flux divergence: $D(\mathbf{x}) \leftarrow \nabla \cdot \mathbf{\Phi}(\mathbf{x})$
			
			\State \textit{// Step 3: Physics-Based Filtering}
			\For{each grid point $\mathbf{x}$}
			\If{$|D(\mathbf{x})| > \beta \|\boldsymbol{\omega}(\mathbf{x})\|$} 
			\State $\mathcal{H}(\mathbf{x}) \leftarrow 0$ \Comment{Suppress sink/source artifacts}
			\Else
			\State $\mathcal{H}(\mathbf{x}) \leftarrow \|\boldsymbol{\omega}(\mathbf{x})\|$ \Comment{Retain loop anomaly signature}
			\EndIf
			\EndFor
			
			\State \Return $\mathcal{H}(\mathbf{x})$
		\end{algorithmic}
	\end{algorithm}

	\begin{table}[!b]
		\vspace*{-2mm}
		\scriptsize
		\caption{Default Simulation Parameters}
		\vspace*{-2mm}
		\label{tab:sim_params}
		\centering
		\renewcommand{\arraystretch}{1.2}
		\begin{tabular}{l c l}
			\toprule
			\textbf{Parameter} & \textbf{Symbol} & \textbf{Default Value} \\
			\midrule
			Network Domain & $\Omega$ & $1000 \times 1000$ m$^2$ \\
			Node Density & $\lambda$ & $2000$ nodes/km$^2$ \\
			Comm. Range & $R_{\rm c}$ & $180$ m \\
			Queue Capacity & $Q_{\max}$ & $40$ packets \\
			Time-to-Live & $T_{\max}$ & $64$ hops \\
			Sink Location & $\mathbf{x}_{\text{sink}}$ & $(500, 500)$ (Center) \\
			Traffic Generation & $\lambda_{\text{pkt}}$ & Poisson Process \\
			Max Gradient Weight & $\alpha_{\max}$ & $0.95$ \\
			Congestion Shaping Factor & $k$ & $3$ \\
			{\color{black}Adaptive EMA Factors} & {\color{black}$\gamma_{\min},\gamma_{\max},\theta$} & {\color{black}$0.1,\ 0.5,\ 1$} \\
			Bloom Filter Size & $M_{\text{BF}}$ & {\color{black}$640$} bits \\
			Number of Bloom hash functions& $k_{\rm BF}$ & $7$ \\
			Grid Resolution & $M \times M$ & $128 \times 128$ \\
			{\color{black}KDE Bandwidth} & $h$ & $90$ m $(\beta_h=0.5)$, default \\
			Filter Threshold & $\beta$ & $1.0$ \\
			\bottomrule
		\end{tabular}
	\end{table}

	{\color{black}
		\subsubsection*{Deployment overhead and parameter sensitivity}
		The discrete VDR implementation mainly relies on local one-hop operations. At node $i$, potential relaxation and flux projection both require $O(d_i)$ operations per update, where $d_i$ is the one-hop degree. Packet forwarding requires $O(k_{\rm BF}d_i)$ operations because each candidate next hop is checked by the Bloom filter and evaluated by the normalized utility. The anomaly filter is executed at the management plane on the reconstructed $M\times M$ field, with $O(M^2)$ grid-level complexity.
		
		In terms of communication overhead, VDR does not require network-wide flooding. Potential values are broadcast only to one-hop neighbors when the event-trigger condition is satisfied, and packet-level loop prevention uses a fixed-size packet-carried Bloom filter. The main parameters also have clear operational roles: $\zeta$ controls the convergence speed of flux projection; $\alpha_{\max}$ and $k$ determine the transition between potential-guided forwarding and queue-pressure response; and $\beta$ controls the sensitivity of circulation-based anomaly filtering. Larger $\alpha_i$ gives more weight to potential descent, whereas smaller $\alpha_i$ increases the relative influence of queue pressure.
	}

	\vspace*{-2mm}
	\section{Numerical Results}\label{S5}
	
	In this section, we evaluate the performance of the proposed VDR framework through numerical simulations. The experiments are designed to validate the theoretical foundations, assess the diagnostic capabilities, and quantify performance gains in dynamic routing scenarios. Unless otherwise specified, the simulation parameters are detailed in Table \ref{tab:sim_params}.
	\vspace*{-2mm}
	\subsection{Macroscopic Diagnostic and Visual Verification}\label{S5.1}
	
	This section establishes the theoretical baseline of the proposed framework. We verify two fundamental hypotheses derived from Theorem \ref{thm:traffic_decomp}: 1) discrete routing loops correspond mathematically to solenoidal vortices in the continuum limit; and 2) the HHD ensures strict orthogonality, rendering the detection metric immune to normal irrotational traffic flows.
	
	\begin{figure}[!h]
		\vspace*{-2mm}
		\centering
		\begin{minipage}[t]{0.494\linewidth} %
			\centering
			\subfigure[Ideal irrotational flow] {	
				\centering
				\includegraphics[width=1\linewidth]{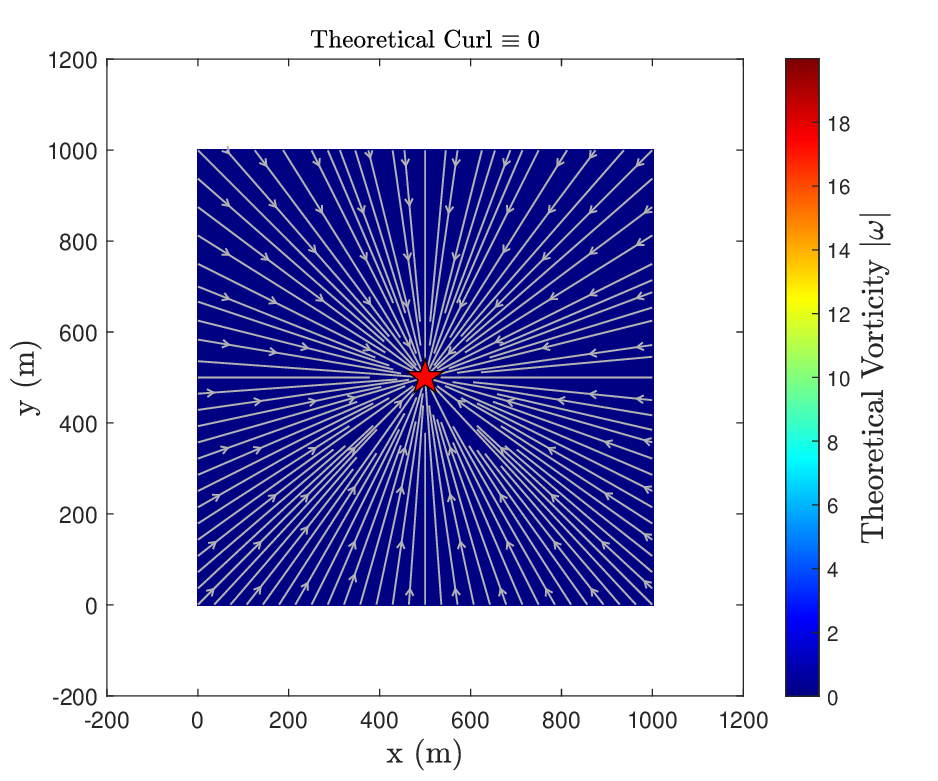}
				\label{fig:case_a}
			}
		\end{minipage} %
		\hspace*{-3mm}
		\begin{minipage}[t]{0.494\linewidth}
			\centering
			\subfigure[Potential flow with a Rankine vortex] {
				\centering
				\includegraphics[width=1\linewidth]{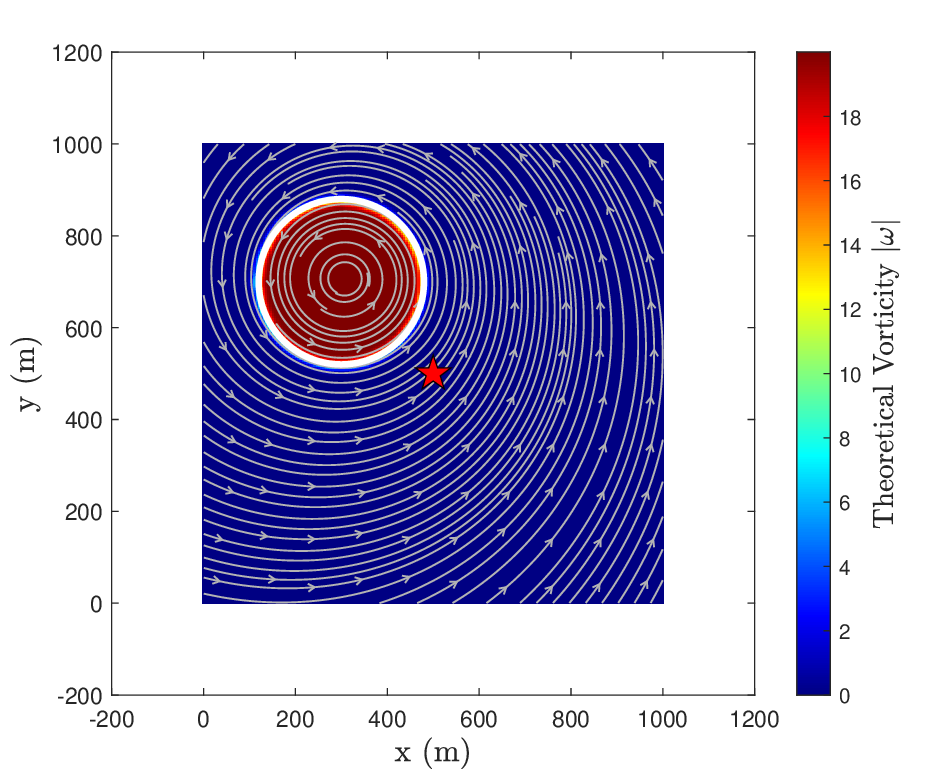}
				\label{fig:case_b}
			}
		\end{minipage}
		\vspace*{-2mm}
		\caption{{\color{black}Visual interpretation of network vorticity using the Helmholtz decomposition.
				A normal irrotational source-to-sink flow produces almost zero vorticity, even when the streamlines converge strongly near the sink.
				In contrast, a routing loop appears as a localized high-vorticity region, indicating loop-induced circulation in the traffic field.}}
	\label{fig:visual_verification} 
	\vspace*{-1mm}
\end{figure}

\subsubsection{Theoretical Validation via Rankine Vortex Model} 
We construct an analytical continuum model within a $1000 \times 1000$\,m$^2$ domain containing two distinct regimes: a pure potential flow converging to a sink and a superimposed Rankine vortex \cite{Saffman1992} centered at $(300, 700)$. Fig. \ref{fig:visual_verification} visualizes the theoretical streamlines and vorticity magnitude. In Fig. \ref{fig:visual_verification}(a), the background vorticity remains uniformly zero despite the high-intensity convergent flux. This confirms the strict orthogonality of the HHD, ensuring that dense multipath routing does not induce spurious vorticity artifacts. Conversely, Fig. \ref{fig:visual_verification}(b) reveals a dual structure: the routing loop manifests as a high-intensity disk (forced vortex) while the surrounding flow remains zero-vorticity (free vortex). The sharp boundary demonstrates the network vorticity metric's capability to precisely isolate the true vortex core, effectively distinguishing the causal anomaly from the surrounding passive curvilinear flows. {\color{black}Visually, the vorticity heatmap should be interpreted as a map of local rotational tendency: dark regions correspond to normal potential-driven transport, whereas bright regions indicate circulating traffic that may correspond to routing loops.}

\begin{figure}[!t]
	\vspace*{-1mm}
	\begin{minipage}[t]{0.494\linewidth} %
		\centering
		\subfigure[Raw discrete routing traces]{
			\label{fig:diagnosis_chaos}
			\vspace{-4mm}
			\includegraphics[width=1\linewidth]{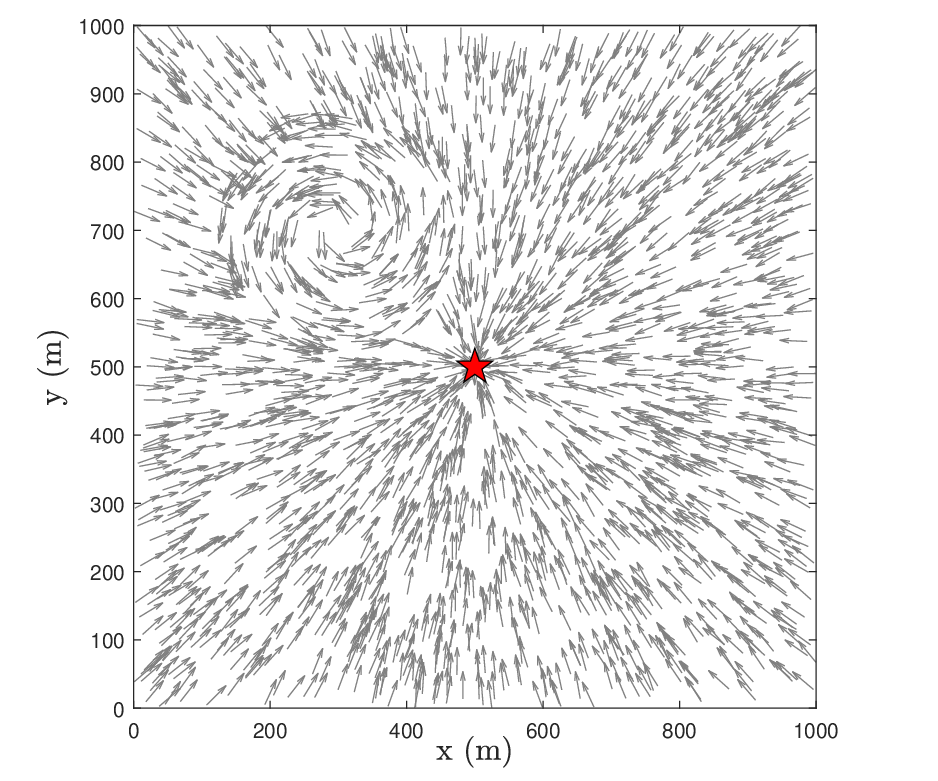}
	}\end{minipage} %
	\hspace*{-3mm}
	\begin{minipage}[t]{0.494\linewidth} %
		\subfigure[Filtered vorticity heatmap]{
			\label{fig:diagnosis_order}
			\includegraphics[width=1\linewidth]{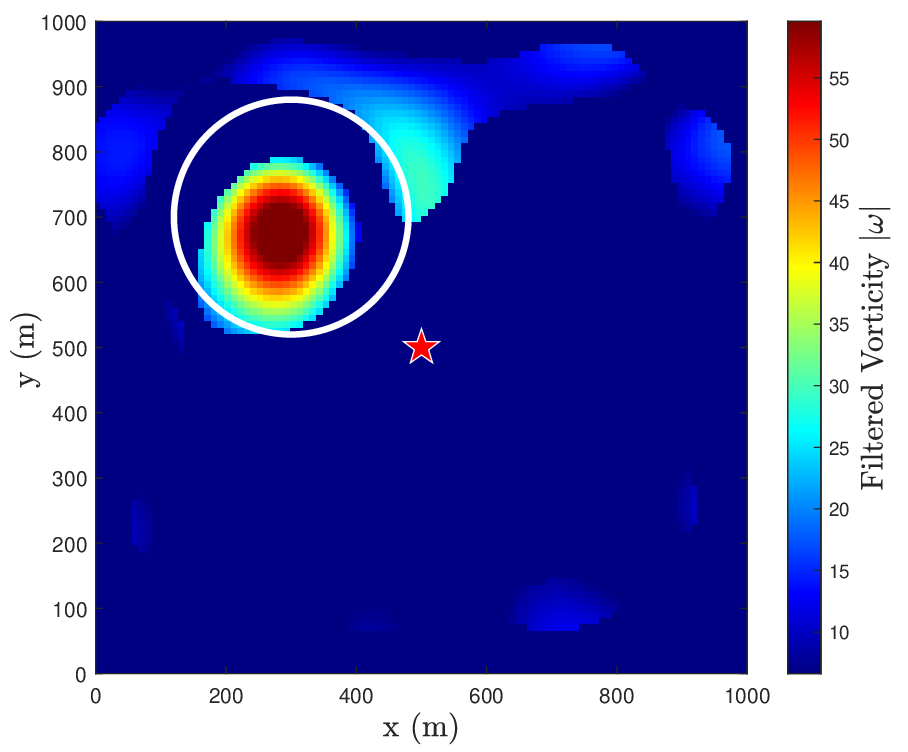}
	}\end{minipage}  
	\vspace*{-2mm}
	\caption{{\color{black}Blind diagnosis of routing anomalies from noisy discrete routing traces.
			The raw forwarding trajectories mix normal sink-directed traffic with anomalous cyclic paths, making the loop difficult to identify directly.
			After divergence-based filtering, the vorticity heatmap suppresses sink/source convergence artifacts and highlights the genuine loop-induced circulation region.}}
\label{fig:discrete_diagnosis}  
\vspace*{-2mm}
\end{figure}

\subsubsection{Blind Diagnosis from Noisy Discrete Traces} 
We further evaluate diagnostic robustness using raw discrete routing logs from a large-scale simulation. As shown in Fig. \ref{fig:discrete_diagnosis}(a), the latent loop is visually obscured by the clutter of background traffic converging toward the sink. Direct application of discrete curl operators on such data typically yields false positives due to shear singularities near the sink. To resolve this, we apply the divergence-based filtering strategy (Algorithm \ref{alg:diagnosis}). By exploiting the kinematic distinction that routing loops are solenoidal while normal flows are irrotational, the filter effectively attenuates high-divergence noise. The resulting heatmap in Fig. \ref{fig:discrete_diagnosis}(b) achieves a high signal-to-noise ratio, suppressing normal traffic into the background while isolating the anomaly. {\color{black}Therefore, high intensity in the filtered vorticity heatmap does not simply indicate high traffic density; rather, it indicates circulation-dominated traffic whose local forwarding directions form a loop-like pattern.} This confirms the framework's capability for blind diagnosis without prior knowledge of the global network state.

\vspace*{-2mm}
\subsection{Sensitivity to KDE Bandwidth}
\label{subsec:bandwidth_sensitivity}

{\color{black}
To validate the spatial abstraction induced by KDE smoothing, we evaluate the sensitivity of vorticity localization to the normalized bandwidth $\beta_h=h/R_c$. Specifically, $\beta_h$ is varied over ${0.25,0.33,0.50,0.67,0.75,1.00}$ while keeping $R_c=180~\mathrm{m}$ fixed. This controlled test keeps the underlying forwarding traces unchanged and isolates the effect of the KDE reconstruction bandwidth on the recovered vorticity field. For localization fidelity, we define the vorticity contrast as
\begin{equation}
	C_\omega =
	\frac{\overline{|\omega|}_{\mathrm{loop}}}
	{\overline{|\omega|}_{\mathrm{bg}}+\epsilon},
	\label{eq:vorticity_contrast}
\end{equation}
where $\overline{|\omega|}_{\mathrm{loop}}$ and $\overline{|\omega|}_{\mathrm{bg}}$ denote the average vorticity magnitude inside the loop region and in the background region, respectively, and $\epsilon$ is a small regularization constant. A larger $C_\omega$ indicates clearer separation between loop-induced circulation and background traffic. The peak preservation ratio is defined as the ratio between the recovered peak vorticity under bandwidth $\beta_h$ and the maximum recovered peak over all tested bandwidths. The localization error is the Euclidean distance between the detected vorticity peak and the ground-truth loop center.
}

\begin{table}[t]
	\vspace*{-3mm}
\scriptsize
\centering
\caption{{\color{black}Sensitivity of vorticity localization to the KDE bandwidth $h=\beta_h R_c$.}}
\vspace*{-3mm}
\label{tab:bandwidth_sensitivity}
\begin{tabular}{c|c|c|c|c}
	\hline
	$\beta_h$ & $h$ (m) & $C_\omega$ & Peak Ratio & Loc. Error (m) \\
	\hline
	0.25 & 45  & 24.53 & 0.709 & 52.77 \\
	0.33 & 59.4 & 21.56 & 1.000 & 8.17 \\
	0.50 & 90  & 12.03 & 0.979 & 8.17 \\
	0.67 & 120.6 & 8.20 & 0.716 & 8.17 \\
	0.75 & 135 & 7.57 & 0.611 & 7.92 \\
	1.00 & 180 & 7.16 & 0.384 & 7.92 \\
	\hline
\end{tabular}
\vspace*{-2mm}
\end{table}

{\color{black}
As shown in Table~\ref{tab:bandwidth_sensitivity}, the KDE bandwidth exhibits a clear resolution--smoothness tradeoff. A small bandwidth, e.g., $\beta_h=0.25$, yields a high vorticity contrast but also leads to a large localization error, indicating sensitivity to local sampling fluctuations. In contrast, overly large bandwidths reduce both the vorticity peak ratio and the contrast, showing that excessive smoothing may diffuse localized anomalies. The default value $\beta_h=0.5$ preserves a low localization error while maintaining a high peak ratio, and is therefore used as a stable compromise between noise suppression and anomaly localization in the considered scenario.
}

\vspace*{-2mm}
\subsection{Dynamic Evolution and Convergence Analysis}\label{S5.2}

We investigate the stability of the control mechanism by simulating a mixed flow state where a Rankine vortex is superimposed on a background potential flow. Fig. \ref{fig:convergence} tracks the energy trajectories over 60 iterations. The vorticity energy $\mathcal{E}_{\mathrm{vor}}$ (blue curve, left logarithmic axis) exhibits an approximately linear decline on the logarithmic scale, indicating that the discrete routing loop follows an exponential decay law ($\mathcal{E}_{\mathrm{vor}} \propto (1-\zeta)^{2\tau}$), consistent with our Lyapunov analysis. Quantitatively, the loop intensity is suppressed by over five orders of magnitude ($10^3 \to 10^{-2}$), confirming rapid anomaly elimination. Simultaneously, the transport energy $\mathcal{E}_{\mathrm{irr}}$ (red curve, right linear axis) remains quasi-constant with negligible variation ($<0.1\%$). This verifies the orthogonality of the discrete Helmholtz decomposition: the dissipation operator selectively targets the solenoidal subspace without perturbing the irrotational potential field, thereby preserving normal packet transmission. The alignment between the discrete trajectory and the theoretical exponential law further supports the framework's structural stability against discretization errors. {\color{black}
Therefore, the 60 iterations in Fig.~\ref{fig:convergence} should be interpreted as a convergence visualization from a mixed initial field, rather than as 60 mandatory network-wide synchronization rounds before packet forwarding.
}

\begin{figure}[!h]
\vspace*{-4mm}
\centering
\includegraphics[width=0.85\linewidth]{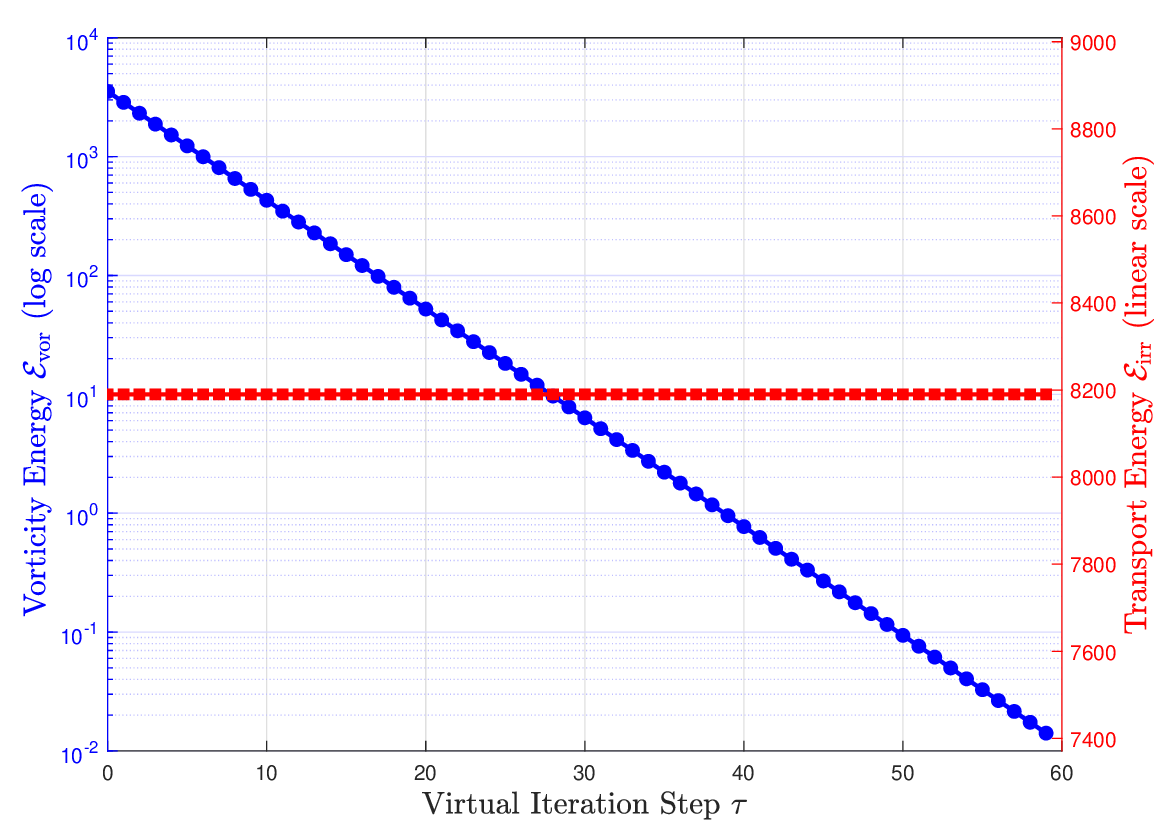} 
\vspace*{-4mm}
\caption{{\color{black}Dynamic evolution of network energy components over virtual relaxation iterations. The horizontal axis denotes the algorithmic index $\tau$, not per-packet physical signalling rounds; in practice, potential-field updates are amortized over the control-plane refresh window.}}
\label{fig:convergence} 
\vspace*{-4mm}
\end{figure}

\vspace*{-2mm}
\subsection{System Performance Comparison: Macro-Benchmarks}\label{S5.3}

To assess routing performance in complex, unstructured topologies characterized by high path diversity and prevalent routing loops, we extend the evaluation to a large-scale random geometric graph scenario. To simulate a dense IoT mesh network, $N$ nodes are randomly distributed over a $1000 \times 1000$ m$^2$ area. In this environment, we compare the proposed VDR against four representative baselines. These comprise two static routing protocols: a global shortest path (SP) strategy, which serves as an idealized representative of proactive link-state protocols like OLSR \cite{Clausen2003OLSR}, and GPSR \cite{Karp2000GPSR} as well as two dynamic, queue-aware protocols: backpressure (BP) \cite{Bui2011} and Q-learning-based QTAR \cite{Arafat2022QL}.

{\color{black}
VDR differs from BP and QTAR in its control principle. BP relies on local queue differentials and may suffer from random-walk-like detours when queue gradients are weak, while QTAR may introduce transient delay due to exploration and policy updates. In contrast, VDR combines a destination-aligned potential drive, adaptive queue-pressure response, and vorticity dissipation. This allows VDR to retain low delay under light traffic while improving congestion adaptability under heavy traffic.
}

\begin{figure}[!b]
\vspace*{-2mm}
\centering
\includegraphics[width=0.85\linewidth]{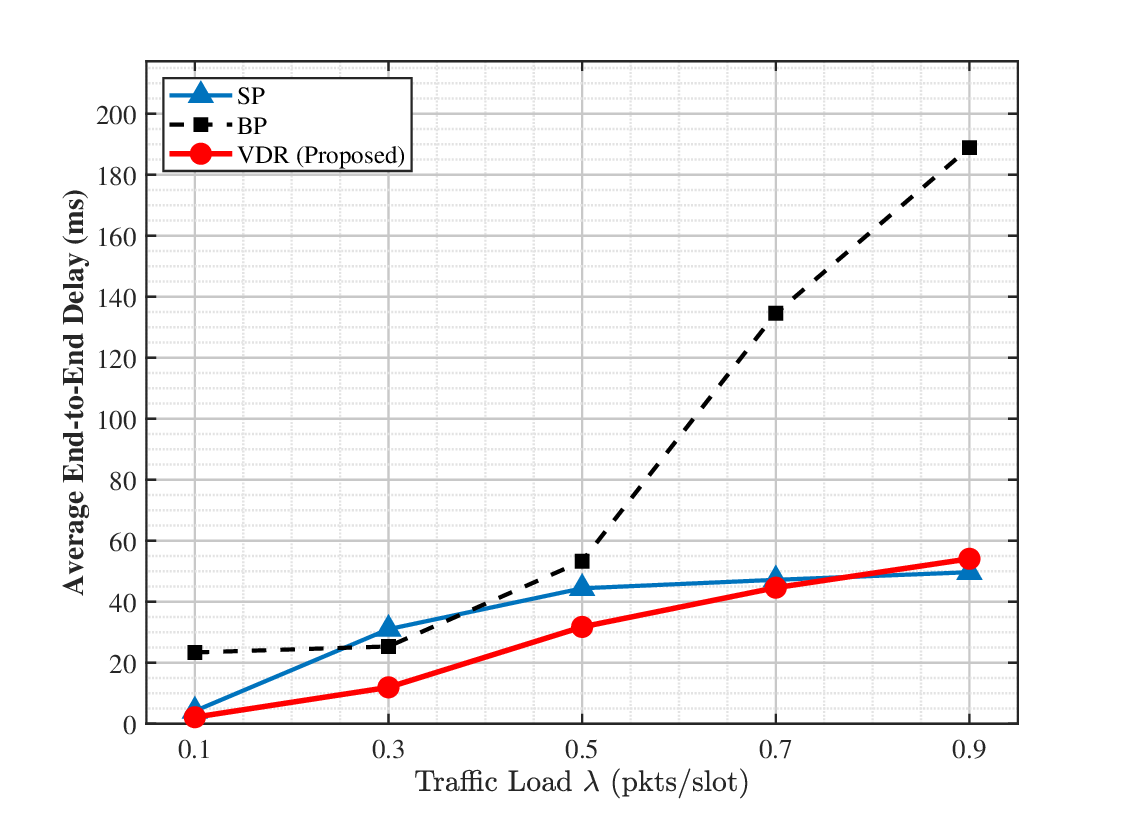} 
\vspace*{-4mm}
\caption{Average end-to-end delay characteristics as a function of traffic load.}
\label{fig:delay_perf} 
\vspace*{-1mm}
\end{figure}

\subsubsection{End-to-End Delay Analysis}\label{S5.3.1}

Fig.~\ref{fig:delay_perf} illustrates the average end-to-end delay of successfully delivered packets as a function of the normalized traffic load. 
{\color{black}
For VDR, if a packet is temporarily held because all admissible next-hop candidates are blocked or have non-positive utility, the holding time is accumulated in its end-to-end delay once the packet is eventually delivered. Packets that expire due to TTL or are dropped due to buffer overflow are not included in the delivered-packet delay average, but are counted as delivery failures in the PDR metric.
}

{\color{black}
In this delay experiment, we focus on SP, BP, and VDR. SP serves as a deterministic shortest-path baseline, while BP represents queue-differential forwarding in which weak or noisy queue gradients may induce random-walk-like detours. Since the delay metric is computed only over successfully delivered packets, the SP curve should be interpreted together with the PDR results in Fig.~\ref{fig:pdr_perf}. In particular, under heavy traffic, shortest-path forwarding may route many packets into bottleneck nodes; packets trapped in these congested regions are more likely to be dropped or remain undelivered and therefore do not enter the delivered-packet delay average.
}
In contrast, BP exhibits rapidly increasing delay as the offered load grows, reflecting the path dilation caused by queue-differential forwarding under local pressure fluctuations. 

{\color{black}
VDR substantially reduces this delay by combining a normalized potential-drive term with a normalized queue-pressure term and a packet-level loop-history guard. Quantitatively, over the entire load range, BP yields an average delay of $85.09$ ms, whereas VDR reduces this value to $28.90$ ms, corresponding to a $66.0\%$ average latency reduction. The point-wise reduction remains substantial under all tested loads: VDR reduces delay by $90.9\%$ at $\lambda=0.1$, $52.9\%$ at $\lambda=0.3$, $40.4\%$ at $\lambda=0.5$, $66.8\%$ at $\lambda=0.7$, and $71.4\%$ at $\lambda=0.9$.
}

At low loads, BP suffers from a latency penalty because the lack of strong queue differentials induces inefficient random walks. VDR avoids this issue because the irrotational potential drive supplies a destination-aligned descent direction even when the queue-pressure gradient is weak or temporarily flat. 
{\color{black}
As congestion increases, the adaptive weight gradually increases the influence of the normalized queue-pressure term, allowing VDR to detour packets around local hotspots while retaining the geometric guidance of the potential field. Thus, VDR achieves a predictable latency-reliability trade-off: it remains close to the low-delay behavior of potential-guided forwarding under light traffic and avoids the severe random-walk delay of BP under heavier traffic.
}

\begin{figure}[!b]
\vspace*{-3mm}
\centering
\includegraphics[width=0.85\linewidth]{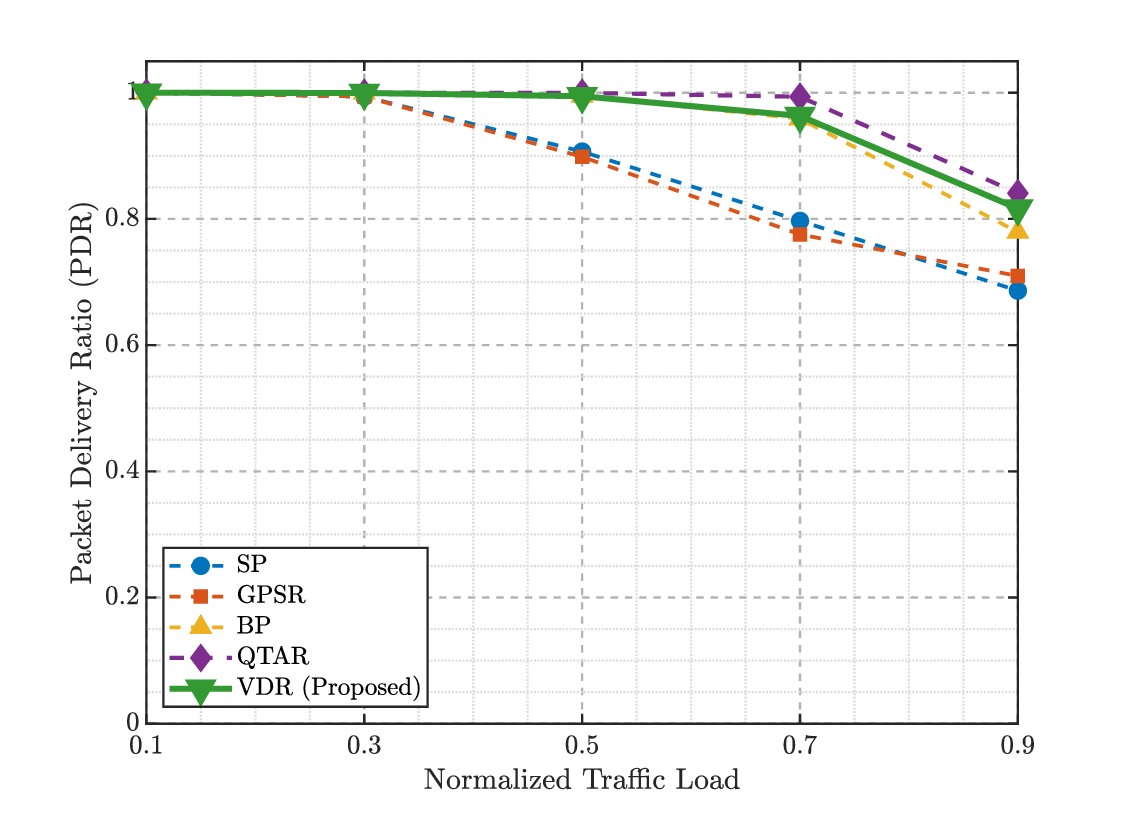} 
\vspace*{-4mm}
\caption{Comparative analysis of packet delivery ratio versus traffic load.}
\label{fig:pdr_perf} 
\vspace*{-1mm}
\end{figure}

\subsubsection{Packet Delivery Ratio Analysis}\label{S5.3.2}

Fig.~\ref{fig:pdr_perf} evaluates network reliability via the packet delivery ratio (PDR). 
{\color{black}
At light traffic loads, all routing schemes maintain a high PDR because congestion and receiver contention are still weak. As the normalized load increases, however, the reliability gap among different routing mechanisms becomes more visible. SP and GPSR degrade more rapidly because their forwarding decisions are mainly determined by shortest-hop or greedy geographic progress, which tends to concentrate traffic around sink-side bottlenecks under heavy load.}
Dynamic protocols mitigate this effect by incorporating congestion- or learning-aware adaptation, but they exhibit different operational trade-offs.

{\color{black}
Quantitatively, at the highest tested load $\lambda=0.9$, SP and GPSR achieve PDR values of $0.6863$ and $0.7097$, respectively, whereas BP improves the PDR to $0.7793$. VDR further increases the PDR to $0.8166$, corresponding to relative gains of approximately $19.0\%$, $15.1\%$, and $4.8\%$ over SP, GPSR, and BP, respectively. Compared with QTAR, whose PDR is $0.8408$ at the same load, VDR remains within a small $2.42$ percentage-point gap.}
This result indicates that VDR achieves reliability comparable to the learning-based baseline while avoiding the exploration overhead and transient instability associated with learning-based routing.

{\color{black}
The improvement of VDR comes from the joint effect of destination-aligned potential descent, adaptive queue-pressure response, and packet-level loop suppression. In the low-load regime, the scalar potential term provides deterministic geometric guidance, preventing the directionless forwarding behavior that may occur in purely local queue-differential routing. Under heavier loads, the queue-pressure term becomes more influential and allows VDR to route packets around congested regions. Therefore, VDR does not simply trade reliability for loop safety: it maintains a robust PDR trajectory while retaining the deterministic loop-suppression property verified in Fig.~\ref{fig:loop_perf}.} {\color{black}
The PDR calculation includes packets lost due to TTL expiration or local buffer overflow after temporary holding; hence, the high PDR of VDR is not obtained by excluding stalled packets from the statistics.
}

\begin{figure}[!h]
\vspace*{-3mm}
\centering
\includegraphics[width=0.85\linewidth]{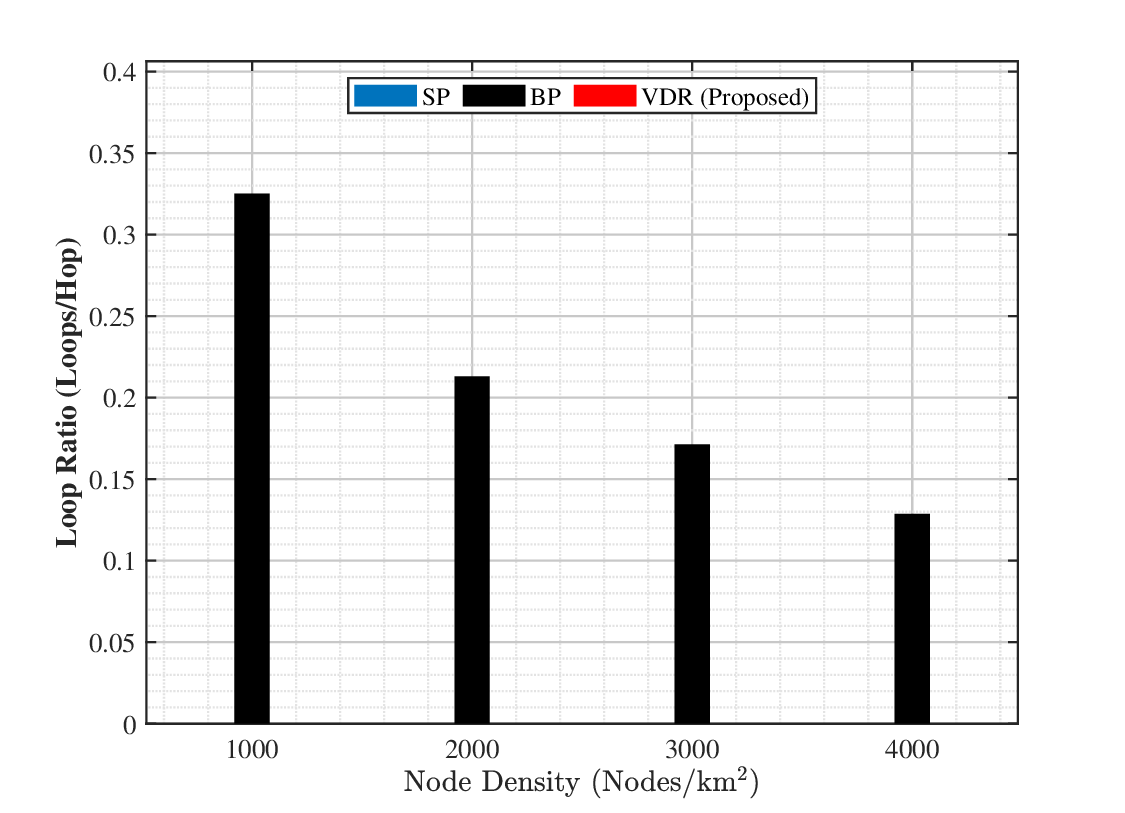} 
\vspace*{-4mm}
\caption{Impact of node density on the realized routing loop ratio. The loop ratio counts completed forwarding cycles, not blocked loop-closure attempts.}
\label{fig:loop_perf} 
\vspace*{-2mm}
\end{figure}

\subsubsection{Loop Suppression Capability}\label{S5.3.3}

Routing loops are a major source of inefficiency in ultra-dense networks, since cyclic forwarding consumes bandwidth and radio resources without contributing to net source-to-sink delivery. Fig.~\ref{fig:loop_perf} evaluates the realized routing loop ratio across different node densities.

{\color{black}
In this loop-suppression experiment, we focus on SP, BP, and VDR. SP serves as a deterministic loop-free reference, BP represents queue-differential forwarding where cyclic forwarding may arise under weak or noisy pressure gradients, and VDR is the proposed loop-suppression method. BP exhibits persistent cyclic forwarding, with the loop ratio decreasing from approximately $0.32$ to $0.13$ loops/hop as the density increases from $1000$ to $4000$ nodes/km$^2$. In contrast, both SP and VDR maintain a zero realized loop ratio across all tested densities.
}

This behavior reflects the limitation of relying solely on localized queue differentials. Without a destination-aligned macroscopic gradient, packets may enter cyclic trajectories when the local pressure field is noisy or temporarily flat. Static shortest-path forwarding avoids loops by following a fixed acyclic route structure, but this rigidity can concentrate traffic around bottleneck nodes, as reflected by the reliability comparison in Fig.~\ref{fig:pdr_perf}.

{\color{black}
VDR suppresses realized forwarding loops through two coupled mechanisms. At the field level, the curl-free potential component provides destination-aligned guidance, while the vorticity dissipation mechanism reduces the solenoidal circulation associated with loop-prone traffic patterns. At the packet level, the trajectory-history guard prevents finite-graph micro-loop closure by blocking candidates that have already appeared in the packet history. Thus, the empirical loop suppression in Fig.~\ref{fig:loop_perf} is consistent with the field-level dissipation principle in Theorem~\ref{thm:enstrophy_dissipation}, while the strict zero realized loop ratio is enforced by the packet-level history guard.
}

{\color{black}
It should be noted that Fig.~\ref{fig:loop_perf} reports realized forwarding loops, not blocked loop-closure attempts. A loop is counted only when a packet is actually forwarded to a node already recorded in its trajectory history. In VDR, such an attempted revisit is intercepted before transmission. The packet is then temporarily held and re-evaluated in the next scheduling slot if no admissible positive-utility candidate is available. The resulting waiting time is included in the end-to-end delay statistics, and any TTL expiration or buffer overflow is counted as packet loss in the PDR metric. Therefore, the zero realized loop ratio in Fig.~\ref{fig:loop_perf} is consistent with the PDR behavior in Fig.~\ref{fig:pdr_perf}: loop closure is prevented at the forwarding decision stage, while the cost of prevention is accounted for through delay and delivery outcomes.
}

\subsubsection{Robustness under Bursty Traffic}\label{S5.3.4}

{\color{black}
To verify the robustness of the source-intensity estimator under transient machine-type traffic, we further evaluate VDR with ON-OFF bursty arrivals. During the ON period, a subset of hotspot sources increases its arrival rate by a burst factor $B$, while the remaining nodes keep the baseline rate. We compare the original fixed EMA with $\gamma=0.1$ and the adaptive EMA in \eqref{eq:adaptive_source_ema}--\eqref{eq:adaptive_gamma}. The packet-level forwarding rule is kept unchanged, so this experiment isolates the effect of source-field tracking.

Table~\ref{tab:bursty_traffic} shows that the adaptive EMA substantially reduces the source-field recovery time under bursty traffic. For example, when the burst factor is $B=10$, the recovery time decreases from $21.75$ update windows with the fixed EMA to $2.80$ update windows with the adaptive EMA. Meanwhile, the PDR remains stable, increasing slightly from $0.9535$ to $0.9548$, and the loop ratio remains zero. Similar behavior is observed for $B=5$, where the recovery time is reduced from $19.97$ to $2.48$ update windows without degrading delivery performance. These results confirm that the adaptive estimator improves transient source-field tracking while preserving the packet-level reliability and loop-suppression behavior of VDR.
}

\begin{table}[!t]
\centering
\caption{{\color{black}Robustness of VDR under bursty traffic.}}
\vspace*{-3mm}
\label{tab:bursty_traffic}
\scriptsize
\setlength{\tabcolsep}{3pt}
\renewcommand{\arraystretch}{1.1}
\begin{tabular}{c|c|c|c|c|c}
	\hline
	$B$ & Estimator & Recovery time & PDR & Avg. delay & Loop ratio \\
	\hline
	1  & Fixed EMA & 0.00 & 0.9957 & 32.36 & 0 \\
	1  & Adaptive EMA & 0.00 & 0.9957 & 32.36 & 0 \\
	5  & Fixed EMA & 19.97 & 0.9895 & 49.30 & 0 \\
	5  & Adaptive EMA & 2.48 & 0.9895 & 49.30 & 0 \\
	10 & Fixed EMA & 21.75 & 0.9535 & 76.90 & 0 \\
	10 & Adaptive EMA & 2.80 & 0.9548 & 76.91 & 0 \\
	\hline
\end{tabular}
\vspace*{-3mm}
\end{table}

\vspace*{-2mm}
\subsection{Scalability and Computational Complexity}\label{S5.4}

\begin{figure}[!t]
\vspace*{-1mm}
\centering
\includegraphics[width=0.85\linewidth]{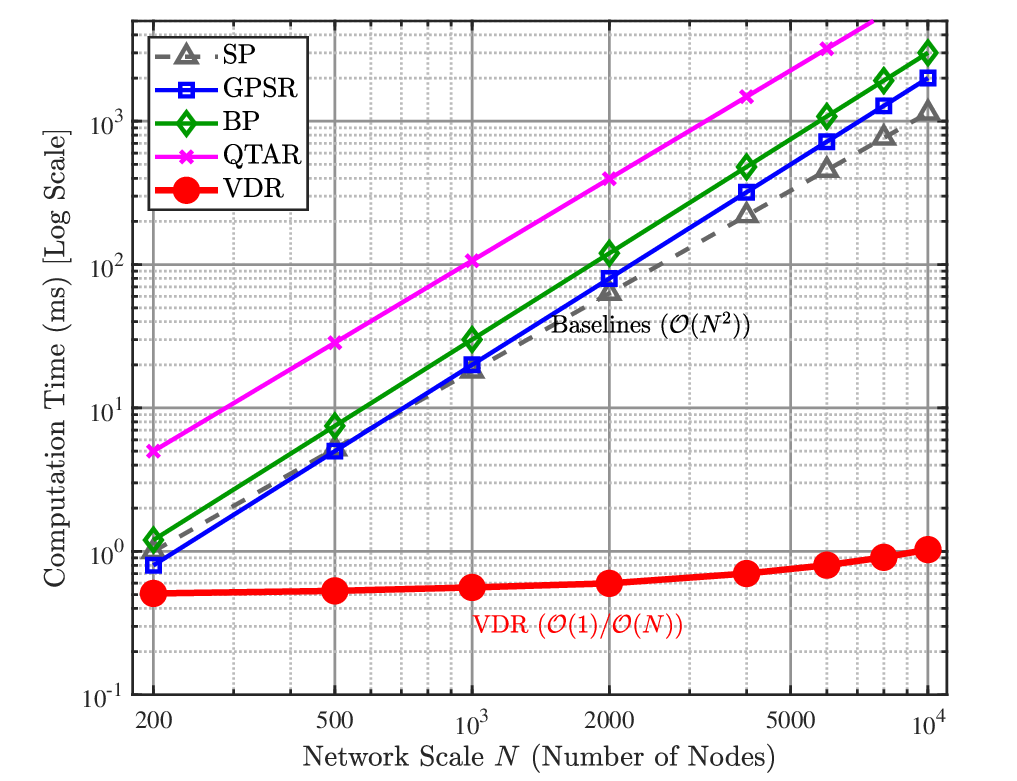}
\vspace*{-4mm}
\caption{{\color{black}	Computational scalability under fixed-area densification. The VDR per-iteration complexity is $O(N+M^2\log M)$. The near-linear trend in $N$ is observed in this experiment because both the physical domain and the grid size $M$ are fixed.}}

\label{fig:scalability} 
\vspace*{-1mm}
\end{figure}

To evaluate the scalability of the proposed framework in ultra-dense deployment scenarios, we analyze the average computation time per routing iteration as the network size $N$ increases from $2 \times 10^2$ to $8 \times 10^3$. The results are presented in Fig.~\ref{fig:scalability} using a log-log scale. Traditional discrete algorithms (SP, GPSR, BP, QTAR) exhibit a quadratic growth trend in this fixed-area densification setting. For instance, at $N=8 \times 10^3$, the update latency for QTAR and BP exceeds $10^3$\,ms. This bottleneck arises because the average node degree increases with $N$, causing the edge set cardinality and the associated graph-theoretic computational load to grow rapidly.

{\color{black}
The proposed VDR algorithm has a per-iteration complexity of
$
O(N+M^2\log M),
$
where the $O(N)$ term accounts for mapping discrete traffic samples onto the continuum field and evaluating local forwarding quantities, while the $O(M^2\log M)$ term accounts for solving the Poisson equation on an $M\times M$ grid. In Fig.~\ref{fig:scalability}, the physical domain is fixed to $1000\times1000~\mathrm{m}^2$ and the grid size is fixed at $M=128$. Therefore, the Poisson-solver cost is constant with respect to $N$, and the measured trend is dominated by the node-dependent $O(N)$ component. This explains the near-linear scaling observed in the fixed-area densification experiment.

When the physical network area increases while a fixed spatial resolution $h$ is maintained, the grid size must also increase. For a square domain with side length $L$, we have $M\approx L/h$, and the grid-dependent cost becomes
$
O\!\left(\frac{L^2}{h^2}\log\frac{L}{h}\right).
$
Thus, under fixed node density and fixed spatial resolution, the overall per-iteration complexity scales as $O(N\log N)$ rather than strictly $O(N)$.
}

Therefore, the near-linear trend in Fig.~\ref{fig:scalability} should be interpreted under the fixed-area, fixed-grid setting used in this experiment. Under this setting, VDR avoids repeated graph traversal over a rapidly densifying edge set and achieves a substantial reduction in computation time at high node densities.

\vspace*{-2mm}
\section{Conclusion}\label{S6}

This work has established a continuum routing framework based on Helmholtz-Hodge decomposition to address scalability limitations in ultra-dense networks. We have proved that macroscopic traffic flux allows for an orthogonal decoupling into demand-driven irrotational and loop-induced solenoidal components. The proposed VDR strategy formulates loop elimination as a gradient flow on the network enstrophy functional, guaranteeing the monotonic decay of routing vorticity. Numerical simulations validated that the framework achieves asymptotic loop-free routing and exhibits near-linear computational scaling under fixed-area, fixed-grid densification, while explicitly accounting for the grid-dependent Poisson-solver cost. Future research will extend this fluid-kinetic paradigm to three-dimensional vehicular networks and cognitive radio environments.




\small

\begin{thebibliography}{10}
\providecommand{\url}[1]{#1}
\csname url@samestyle\endcsname
\providecommand{\newblock}{\relax}
\providecommand{\bibinfo}[2]{#2}
\providecommand{\BIBentrySTDinterwordspacing}{\spaceskip=0pt\relax}
\providecommand{\BIBentryALTinterwordstretchfactor}{4}
\providecommand{\BIBentryALTinterwordspacing}{\spaceskip=\fontdimen2\font plus
	\BIBentryALTinterwordstretchfactor\fontdimen3\font minus
	\fontdimen4\font\relax}
\providecommand{\BIBforeignlanguage}[2]{{%
		\expandafter\ifx\csname l@#1\endcsname\relax
		\typeout{** WARNING: IEEEtran.bst: No hyphenation pattern has been}%
		\typeout{** loaded for the language `#1'. Using the pattern for}%
		\typeout{** the default language instead.}%
		\else
		\language=\csname l@#1\endcsname
		\fi
		#2}}
\providecommand{\BIBdecl}{\relax}   
\BIBdecl

\bibitem{IoTJDong} 
W.-Y. Dong, S. Yang, P. Zhang, and S. Chen, ``Modeling and performance analysis of IoT-over-LEO satellite systems under realistic operational constraints: A stochastic geometry approach,'' \emph{IEEE Internet Things J.}, vol. 12, no. 15, pp. 30576--30593, Aug. 2025. 

\bibitem{Polese2020} 
M. Polese, \emph{et al.}, ``Integrated access and backhaul in 5G mmWave networks: Potential and challenges,'' \emph{IEEE Commun. Mag.}, vol. 58, no. 3, pp. 62--68, Mar. 2020.

\bibitem{Kim2017RPLSurvey} 
H.-S. Kim, J. Ko, D. E. Culler, and J. Paek, ``Challenging the IPv6 routing protocol for low-power and lossy networks (RPL): A survey,'' \textit{IEEE Commun. Surv. Tuts.}, vol. 19, no. 4, pp. 2502--2525, 4th Quart. 2017.

\bibitem{Lighthill1955} 
M.~J. Lighthill and G.~B. Whitham, ``On kinematic waves II. A theory of traffic flow on long crowded roads,'' \emph{Proc. Roy. Soc. London. Ser. A, Math. Phys. Sci.}, vol. 229, no. 1178, pp. 317--345, May 1955.


\bibitem{Griffiths1999} 
D.~J. Griffiths, \emph{Introduction to Electrodynamics} (3rd~ed.). Upper Saddle River, NJ, USA: Prentice Hall, 1999.


\bibitem{Perkins1999AODV} 
C. E. Perkins and E. M. Royer, ``Ad-hoc on-demand distance vector routing,'' in \emph{Proc. WMCSA'99} (New Orleans, LA, USA), Feb.~25-26, 1999, pp. 90--100.

\bibitem{Clausen2003OLSR} 
T. Clausen and P. Jacquet, ``Optimized Link State Routing Protocol (OLSR),'' Internet Engineering Task Force (IETF), RFC 3626, Oct. 2003, doi: 10.17487/RFC3626.


\bibitem{Karp2000GPSR} 
B. Karp and H. T. Kung, ``GPSR: Greedy perimeter stateless routing for wireless networks,'' in \textit{Proc. MobiCom'00} (Boston, MA, USA), Aug.~6-11, 2000, pp. 243--254.


\bibitem{Wang2025TVT} 
C.-M. Wang, \emph{et al.}, ``A distributed hybrid proactive-reactive ant colony routing protocol for highly dynamic FANETs with link quality prediction,'' \emph{IEEE Trans. Veh. Technol.}, vol. 74, no. 1, pp. 1817--1822, Jan. 2025.

\bibitem{Yang2025TVT} 
S. Yang, \emph{et al.}, ``Betweenness centrality based dynamic source routing for flying ad hoc networks in marching formation,'' \emph{IEEE Trans. Veh. Technol.}, vol. 74, no. 8, pp. 12791--12798, Aug. 2025.

\bibitem{Ding2026ISACWaveform}
X. Ding \emph{et al.}, ``Integrated communication and sensing waveform in low-altitude wireless network,'' \emph{IEEE Commun. Stand. Mag.}, vol.~10, no.~2, pp.~140--149, Jun.~2026.


\bibitem{Tassiulas1992} 
L. Tassiulas and A. Ephremides, ``Stability properties of constrained queueing systems and scheduling policies for maximum throughput in multihop radio networks,'' \emph{IEEE Trans. Autom. Control}, vol. 37, no. 12, pp. 1936--1948, Dec. 1992.

\bibitem{Bui2011} 
L.~X.~Bui, R.~Srikant, and A.~Stolyar, ``A novel architecture for reduction of delay and queueing structure complexity in the back-pressure algorithm,'' \emph{IEEE/ACM Trans. Netw.}, vol. 19, no. 6, pp. 1597--1609, Dec. 2011.

\bibitem{Deng2023DBPR} 
X. Deng, \emph{et al.}, ``Distance-based back-pressure routing for load-balancing LEO satellite networks,'' \emph{IEEE Trans. Veh. Technol.}, vol. 72, no. 1, pp. 1240--1253, Jan. 2023.

\bibitem{Zhao2024BiasedBP} 
Z. Zhao, \emph{et al.}, ``Biased backpressure routing using link features and graph neural networks,'' \emph{IEEE Trans. Mach. Learn. Commun. Netw.}, vol. 2, pp. 1424--1439, Oct. 2024.

\bibitem{Mahfujul2024TNSE} 
K. Mahfujul, K. Qu, Q. Ye, and N. Lu, ``Augmenting backpressure scheduling and routing for wireless computing networks,'' \emph{IEEE Trans. Netw. Sci. Eng.}, vol. 11, no. 6, pp. 6605--6622, Nov.-Dec. 2024.

\bibitem{Arafat2022QL} 
M.~Y.~Arafat and S.~Moh, ``A Q-learning-based topology-aware routing protocol for flying ad hoc networks,'' \textit{IEEE Internet Things J.}, vol.~9, no.~3, pp.~1985--2000, Feb.~2022.

\bibitem{Xiao2021DRLSurvey} 
Y. Xiao, J. Liu, J. Wu, and N. Ansari, ``Leveraging deep reinforcement learning for traffic engineering: A survey,'' \emph{IEEE Commun. Surveys Tuts.}, vol. 23, no. 4, pp. 2064--2097, 4th Quart. 2021.

\bibitem{Rusek2020} 
K.~Rusek, \emph{et al.}, ``RouteNet: Leveraging graph neural networks for network modeling and optimization in SDN,'' \emph{IEEE J. Sel. Areas Commun.}, vol. 38, no. 10, pp. 2260--2270, Oct. 2020.

\bibitem{Shen2021GNN} 
Y. Shen, Y. Shi, J. Zhang, and K. B. Letaief, ``Graph neural networks for scalable radio resource management: Architecture design and theoretical analysis,'' \emph{IEEE J. Sel. Areas Commun.}, vol. 39, no. 1, pp. 101--115, Jan. 2021.

\bibitem{Khan2024} 
N. Khan, \emph{et al.}, ``Explainable and robust artificial intelligence for trustworthy resource management in 6G networks,'' \emph{IEEE Commun. Mag.}, vol. 62, no. 4, pp. 50--56, Apr. 2024.

\bibitem{Barbarossa2020TSP} 
S. Barbarossa and S. Sardellitti, ``Topological signal processing over simplicial complexes,'' \emph{IEEE Trans. Signal Process.}, vol. 68, pp. 2992--3007, Jun.. 2020.

\bibitem{Schaub2020RandomWalks} 
M. T. Schaub, \emph{et al.}, ``Random walks on simplicial complexes and the normalized Hodge 1-Laplacian,'' \emph{SIAM Rev.}, vol. 62, no. 2, pp. 353--391, 2020.

\bibitem{Jiang2011StatisticalRanking} 
X. Jiang, L.-H. Lim, Y. Yao, and Y. Ye, ``Statistical ranking and combinatorial Hodge theory,'' \emph{Math. Program.}, vol. 127, no. 1, pp. 203--244, 2011.

\bibitem{Jacquet2004} 
P.~Jacquet, ``Geometry of information propagation in massively dense ad hoc networks,'' in \emph{Proc. MobiHoc'04} (Tokyo, Japan), May 24-26, 2004, pp. 157--162.

\bibitem{Toumpis2003Continuum} 
S. Toumpis and A. J. Goldsmith, ``Capacity regions for wireless ad hoc networks,'' \emph{IEEE Trans. Wireless Commun.}, vol. 2, no. 4, pp. 736--748, Jul. 2003.

\bibitem{Samarakoon2016MFG} 
S. Samarakoon, \emph{et al.}, ``Ultra-dense small cell networks: Turning density into energy efficiency,'' \emph{IEEE J. Sel. Areas Commun.}, vol. 34, no. 5, pp. 1267--1280, May 2016.

\bibitem{DongTCOM} 
W.-Y. Dong, S. Yang and S. Chen, ``Uplink performance analysis of heterogeneous non-terrestrial networks in harsh environments: A novel stochastic geometry model,'' \emph{IEEE Trans. Commun.}, vol. 73, no. 8, pp. 6734--6747, Aug. 2025.

\bibitem{DongJSAC} 
W.-Y.~Dong, S.~Yang, P.~Zhang, and S.~Chen, ``Stochastic geometry based modeling and analysis of uplink cooperative satellite-aerial-terrestrial networks for nomadic communications with weak satellite coverage,'' \emph{IEEE J. Sel. Areas Commun.}, vol.~42, no.~12, pp.~3428--3444, Dec.~2024.

\bibitem{Silverman1986} 
B. W. Silverman, \emph{Density Estimation for Statistics and Data Analysis}. London, U.K.: Chapman and Hall, 1986.


\bibitem{Saffman1992} 
P. G. Saffman, \textit{Vortex Dynamics}. Cambridge, U.K.: Cambridge University Press, 1992.

\bibitem{Arfken2013} 
G.~B. Arfken, H.~J. Weber, and F.~E. Harris, \emph{Mathematical Methods for Physicists} (7th~ed.). Oxford, U.K.: Academic Press, 2013.

	\bibitem{Dong2026FSTSG} 
W.-Y. Dong, W. Jiang, S. Zhao, Q. Bi, and S. Chen, ``Fluid-spatiotemporal stochastic geometry: Information flow in non-stationary fields,'' arXiv:2607.00616, 2026.

\bibitem{Dong2026TMC} 
W.-Y. Dong, S. Zhao, R.-S. Han, Q. Bi, and S. Chen, ``Digital tides: A fluid-dynamic framework for flux-aware infrastructure provisioning in UAV logistics networks,'' \emph{IEEE Trans. Mobile Comput.}, early access, Apr.~2026, doi: 10.1109/TMC.2026.3688690.



\end{thebibliography}

\makeatletter
\def\@IEEEBIOskipN{0.4\baselineskip}
\makeatother

\begin{IEEEbiography}
[{\includegraphics[width=1in,height=1.25in,clip,keepaspectratio]{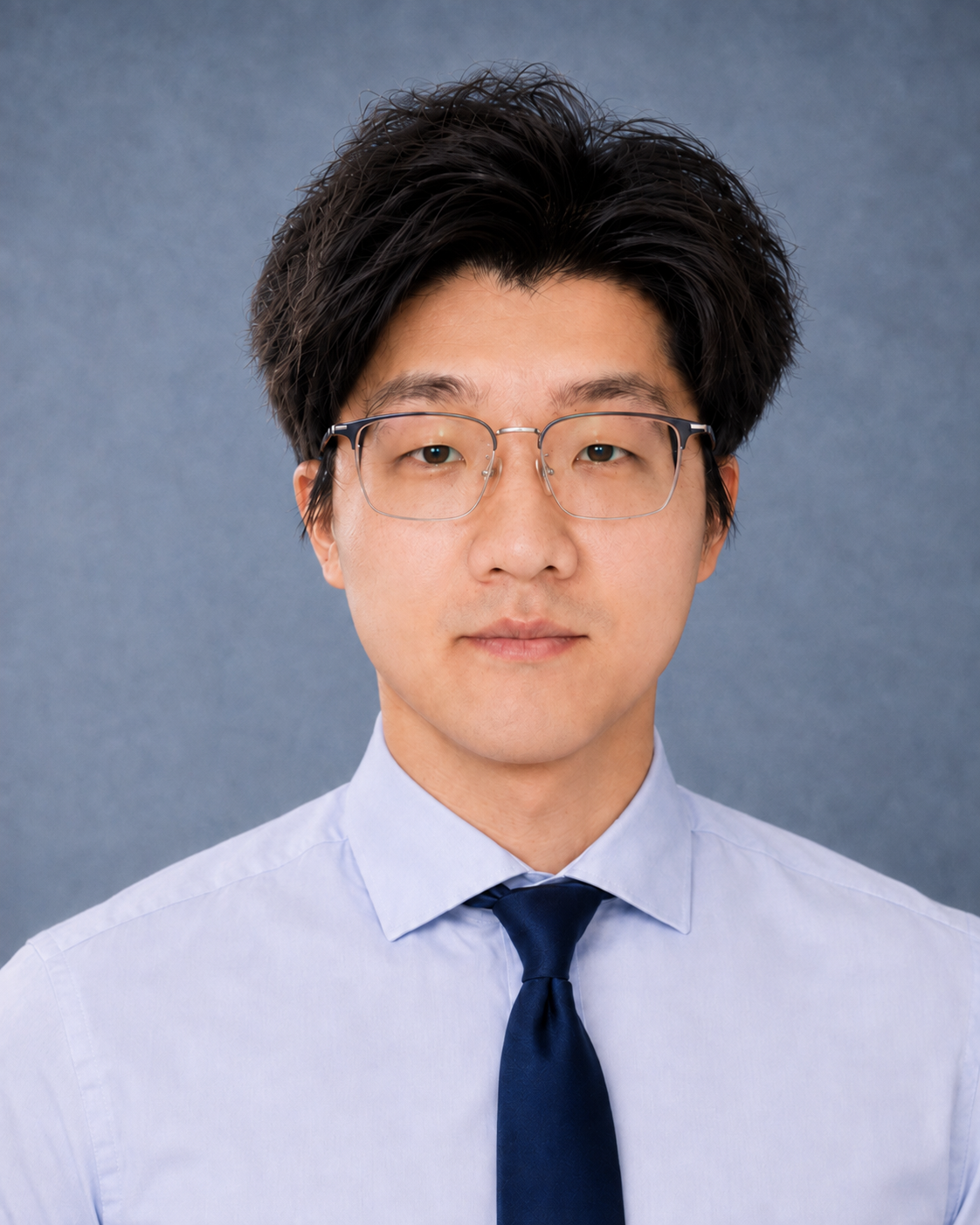}}]
{Wen-Yu Dong}
received the B.S. degree in electronic and information engineering from Sichuan University, Chengdu, China, in 2019, and the Ph.D. degree in information and communication engineering from the School of Information and Communication Engineering, Beijing University of Posts and Telecommunications, Beijing, China, in 2025. He is currently a Researcher with the Future Technology Research Center, China Telecom Research Institute, Beijing, China. His current research interests include space-air-ground integrated networks, fluid-spatiotemporal stochastic geometry (F-STSG), and wireless communications.
\end{IEEEbiography}

\begin{IEEEbiography}
[{\includegraphics[width=1in,height=1.25in,clip,keepaspectratio]{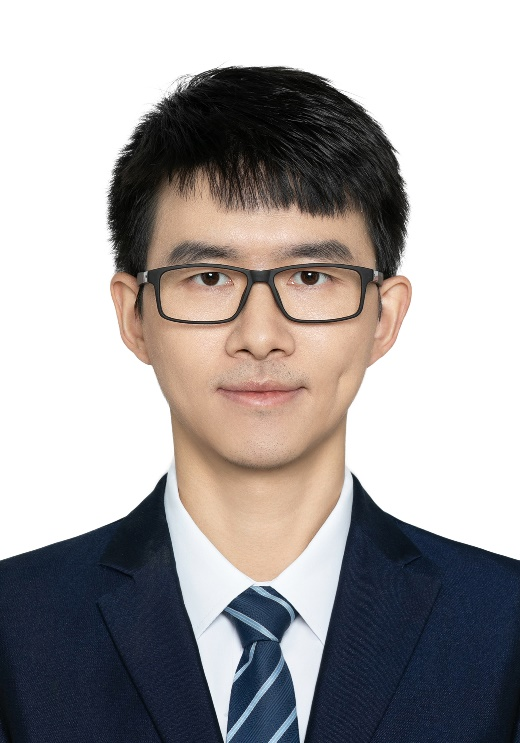}}]
{Weiwei Jiang} (Senior Member, IEEE)
received the B.Sc. Degree of Electronic Engineering and Ph.D. Degree of Information and Communication Engineering from the Department of Electronic Engineering, Tsinghua University, Beijing, China, in 2013 and 2018, respectively. He is currently an associte professor with the School of Information and Communication Engineering, Beijing University of Posts and Telecommunications, and Key Laboratory of Universal Wireless Communications, Ministry of Education. His current research interests include artificial intelligence, machine learning, big data, wireless communication and edge computing. He has published more than 100 academic papers in IEEE Trans and other journals, with more than 5700 citations in Google Scholar. He is one of 2022, 2023, 2024 and 2025 Stanford's List of World's Top 2\% Scientists.
\end{IEEEbiography}

\begin{IEEEbiography}[{\includegraphics[width=1in,height=1.25in,clip,keepaspectratio]{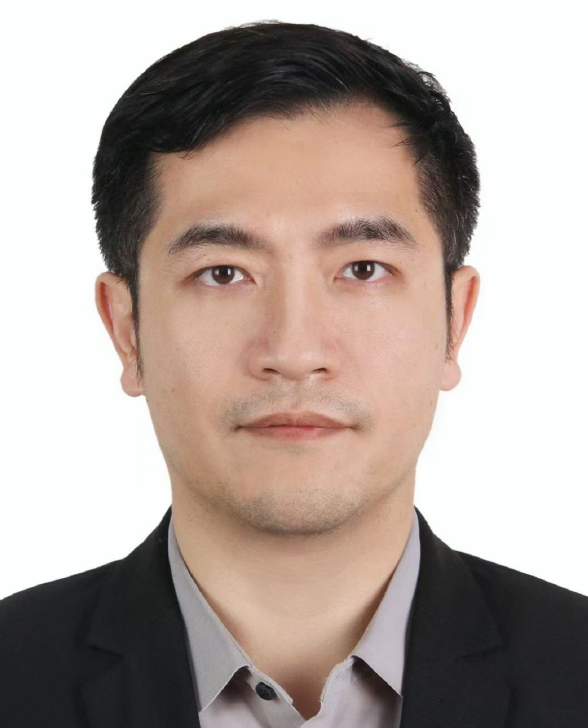}}]{Song Zhao}
is a senior engineer with the Future Technology Center, China Telecom Research Institute, Beijing, China. He received his Ph.D. degree in Telecommunications and Information Systems from Beijing University of Posts and Telecommunications, Beijing, China. He serves as a delegate of China Telecom in the 3GPP SA2 and SA5 working groups, and as rapporteur for multiple 3GPP work items related to network automation and intelligence. His research interests include wireless channel modeling, radio resource management, proximity services, and network AI.
\end{IEEEbiography}

\begin{IEEEbiography}
[{\includegraphics[width=1in,height=1.25in,clip,keepaspectratio]{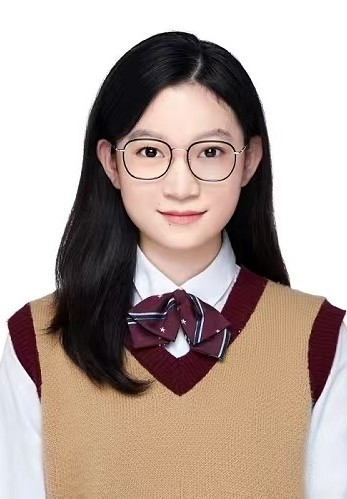}}]
{Rui-Si Han}
received her B.S. degree in E-Commerce and Law from Beijing University of Posts and Telecommunications, Beijing, China, in 2021, and her M.Eng. degree in Information and Communication Engineering from the same university in 2024. She is currently a researcher at the China Telecom Cloud Network Operating System R\&D Center, Beijing, China. Her research interests include mobile ad hoc networks and research project management.
\end{IEEEbiography}

\begin{IEEEbiography}[{\includegraphics[width=1in,height=1.25in,clip,keepaspectratio]{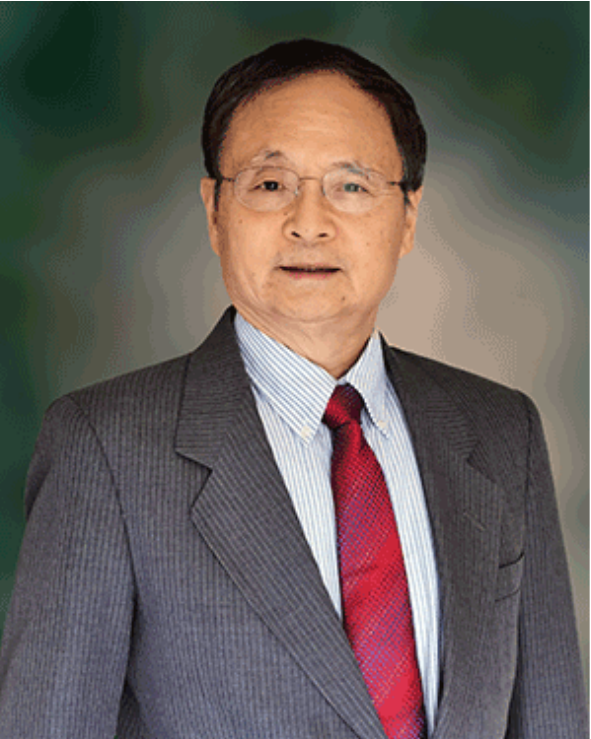}}]{Qi Bi}
(Fellow, IEEE) is Chief Scientist of China Telecom and CTO of its Research Institute, specializing in 5G and 6G. He received his M.S. from Shanghai Jiao Tong University and Ph.D. from Pennsylvania State University. He was awarded Bell Labs Fellow in 2002, the Bell Labs President’s Gold Awards in 2000 and 2002, the Asian American Engineer of the Year in 2005, the Beijing Outstanding Contribution Award for Innovation and Entrepreneurship for Overseas Scholars in 2019, and three Patent Silver Prizes from the China National Intellectual Property Administration from 2022 to 2024.
\end{IEEEbiography}

\begin{IEEEbiography}[{\includegraphics[width=0.9in,height=1.1in,clip,keepaspectratio]{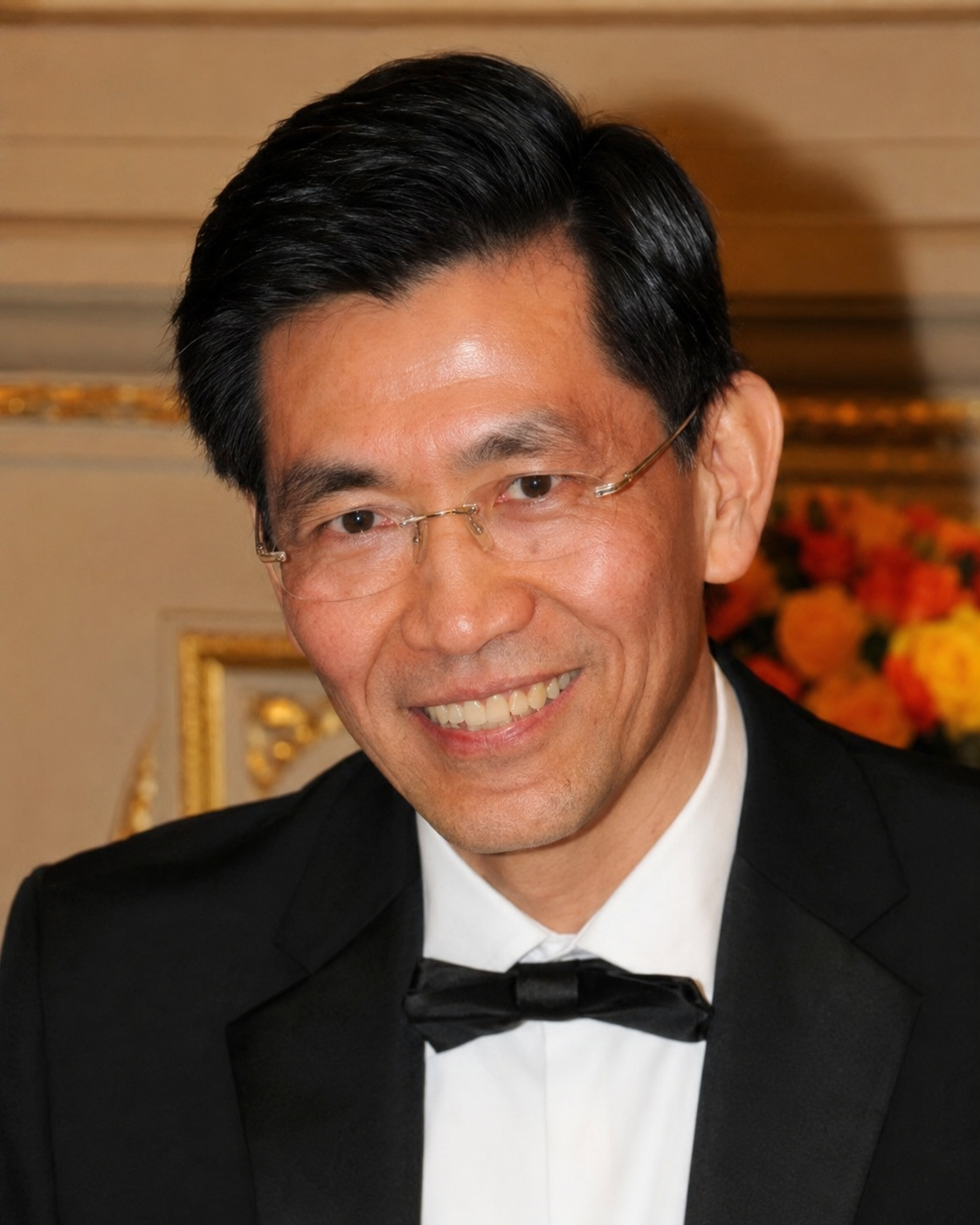}}]{SHENG CHEN} (IEEE Life Fellow) is an Emeritus Professor, University of Southampton, Southampton, U.K.
	
	He received the B.Eng. degree in control engineering from the East China Petroleum Institute, Dongying, China, in January 1982, the Ph.D. degree in control engineering from City University, London, in September 1986, and the higher doctoral (D.Sc.) degree from the University of Southampton, Southampton, U.K., in August 2005. From 1986 to 1999, he held research and academic appointments with the University of Sheffield, the University of Edinburgh, and the University of Portsmouth, all in U.K. In September 1999, he joined the School of Electronics and Computer Science, University of Southampton, and retired in June 2026. Since then, he is a University Visitor associated with the School of Electronics and Computer Science, University of Southampton. In 2023, Dr. Chen was appointed  a Distinguished Professor at Ocean University of China.
	
	Professor Chen has established an over 40-year distinguished academic career in the fields of adaptive signal processing, wireless communications, neural networks and machine learning. He has authored/coauthored over 750 research papers. He has 23,000+ Web of Science citations with h-index 67, and 44,000+ Google Scholar citations with h-index 88. Professor Chen is IEEE ComSoc Signal Processing and Computing for Communications (SPCC) 2025 Technical Recognition Award recipient.
	
	Professor Chen was elected as a Fellow of the United Kingdom Royal Academy of Engineering in 2014. He is a fellow of Asia-Pacific Artificial Intelligence Association, a fellow of Industry Academy of the International Artificial Intelligence Industry Alliance, a fellow of IET, and one of the original 200 ISI Highly Cited Researchers in engineering (March 2004).
\end{IEEEbiography}

\end{document}